\documentclass[twocolumn,superscriptaddress,aps,pra]{revtex4-1}

\usepackage{graphicx}
\usepackage{amsmath}
\usepackage{amsthm}
\usepackage{amssymb}
\usepackage{latexsym}
\usepackage{array}
\usepackage{hyperref}
\usepackage{amsfonts}
\usepackage{dsfont}
\usepackage{mathrsfs}
\usepackage{verbatim}
\usepackage{bbold}
\usepackage[normalem]{ulem}
\usepackage{upgreek}
\usepackage{makecell}
\usepackage{adjustbox}
\usepackage{algorithm}
\usepackage{algpseudocode}
\usepackage{color}
\usepackage{bm}
\usepackage{times}
\usepackage{booktabs}
\usepackage{multirow}

\newcommand{\bra}[1]{\left<#1\right|}
\newcommand{\ket}[1]{\left|#1\right>}
\newcommand{\abs}[1]{\left|#1\right|}
\newcommand{\norm}[1]{\left\lVert#1\right\rVert}

\newcommand{\ketbra}[2]{\ket{#1}\!\!\bra{#2}}
\newcommand{\E}{\mathbb{E}}
\newcommand{\Prob}{\mathbb{P}}
\newcommand{\R}{\mathbb{R}}
\newcommand{\sgn}{\operatorname{sgn}}

\newtheorem{theorem}{Theorem}
\newtheorem{proposition}{Proposition}
\newtheorem{lemma}{Lemma}
\newtheorem{corollary}{Corollary}
\newtheorem{definition}{Definition}

\DeclareMathOperator{\Tr}{Tr}

\DeclareMathOperator{\rank}{rank}

\DeclareMathOperator{\spanop}{span}

\newcommand{\id}{\mathds{1}}

\begin{document}

\title{Learning with Active Quantum Subspaces: Scalable Hybrid Advantage without Full Quantum Data-Encoding}

\author{Jeongho~Bang}\email{jbang@yonsei.ac.kr}
\affiliation{Institute for Convergence Research and Education in Advanced Technology, Yonsei University, Seoul 03722, Republic of Korea}
\affiliation{Department of Quantum Information, Yonsei University, Incheon 21983, Republic of Korea}

\author{Wooyeong~Song}
\affiliation{Quantum Network Research Center, Korea Institute of Science and Technology Information, Daejeon 34141, Korea}

\author{Kyoungho~Cho}
\affiliation{Institute for Convergence Research and Education in Advanced Technology, Yonsei University, Seoul 03722, Republic of Korea}
\affiliation{Department of Statistics and Data Science, Yonsei University, Seoul 03722, Republic of Korea}

\author{Taewan~Kim}
\affiliation{Electronics and Telecommunications Research Institute, Daejeon 34129, Korea}

\author{Yongsoo~Hwang}\email{yhwang@etri.re.kr}
\affiliation{Electronics and Telecommunications Research Institute, Daejeon 34129, Korea}

\date{\today}

\begin{abstract}
We study whether quantum learning advantage can persist without fully embedding a large classical input into a highly superposed quantum state. To address this question, we introduce \emph{active quantum subspace data-encoding}, in which only an information-bearing subset of the input is lifted to a quantum representation while the remaining variables stay classical. For this model, we define a projected hybrid readout and prove three structural results. First, the projected hybrid kernel is positive semidefinite and its sample regularized dimension is bounded by the number of projected observables, so the dimension blow-up of naive global kernels is avoided. Second, we give a necessary and sufficient criterion for improvement over a purely classical predictor in squared loss: the projected quantum sector must contain a direction that lies outside the classical feature span and correlates with the classical residual. Third, in a realizable noisy-oracle setting, we derive a PAC sample-complexity bound proportional to the inverse square of the oracle reliability. We then show, for a canonical Clifford active-subspace family under local dephasing noise, that this reliability can remain inverse-polynomial even when the encoding gate complexity grows polynomially with system size. Hence, the polynomial encoding cost does not by itself destroy the hybrid learning advantage. A sixty-four-qubit family and a synthetic contextual classification task illustrate how one projected quantum feature can compress a useful high-order interaction into a low-dimensional hybrid model. Our results generalize QRAM-free hybrid learning and provide a scalable route toward NISQ-compatible quantum advantage without full quantum data-encoding.
\end{abstract}

\maketitle

\section{Introduction}\label{sec:intro}

Quantum machine learning (QML) has been pursued along at least two complementary lines. One line seeks complexity-theoretic or representation-theoretic quantum advantages by encoding classical data into quantum states and then exploiting a quantum feature map, a variational model, or a quantum kernel~\cite{Biamonte2017,Ciliberto2018,Havlicek2019,SchuldKilloran2019,Cerezo2022,Liu2021NatPhys}. The other line asks a more practical question, especially relevant to noisy intermediate-scale quantum (NISQ) devices~\cite{Preskill2018,Schuld2022PRXQuantum}: can one obtain a tangible learning advantage even when the available quantum system is small and the classical data set is large? A persistent difficulty for the first line is that loading large classical data into a quantum state may itself be expensive and/or fragile. If the cost of the embedding dominates the computation, then the expected quantum learning advantage becomes much less convincing~\cite{Aaronson2015,Paler2020,Harrow2020,Tang2021}.

This tension has become sharper in recent years~\cite{Tang2021,Schuld2022PRXQuantum}. On the one hand, the feature-map view of quantum learning clarifies that the encoding step itself can be regarded as the source of nonlinearity~\cite{Havlicek2019,SchuldKilloran2019,Jerbi2023}. On the other hand, the ``power of data'' perspective shows that hard-to-compute quantum functions do not automatically yield hard-to-learn prediction problems once training data are given~\cite{Huang2021}. In particular, the latter work emphasizes two points that are directly relevant here. First, a quantum model needs to induce data geometry that is genuinely different from what strong classical models already exhibit. Second, a naive global quantum kernel can suffer from an effective-dimension problem: if the embedded points are all too far apart, the representation becomes statistically poor even if it is formally large~\cite{Huang2021}. More recent kernel-concentration analyses sharpen this caution: for sufficiently global or expressive quantum kernels, the kernel values can concentrate exponentially, so polynomial-shot estimates may lead to nearly input-independent predictors~\cite{Thanasilp2024,Agliardi2026}. These observations make it clear that one should not identify quantum learning advantage with ``use more qubits and larger Hilbert spaces''~\cite{Song2021,Song2022,Caro2022,Yin2025}.

\begin{figure*}[t]
\centering
\includegraphics[width=0.75\textwidth]{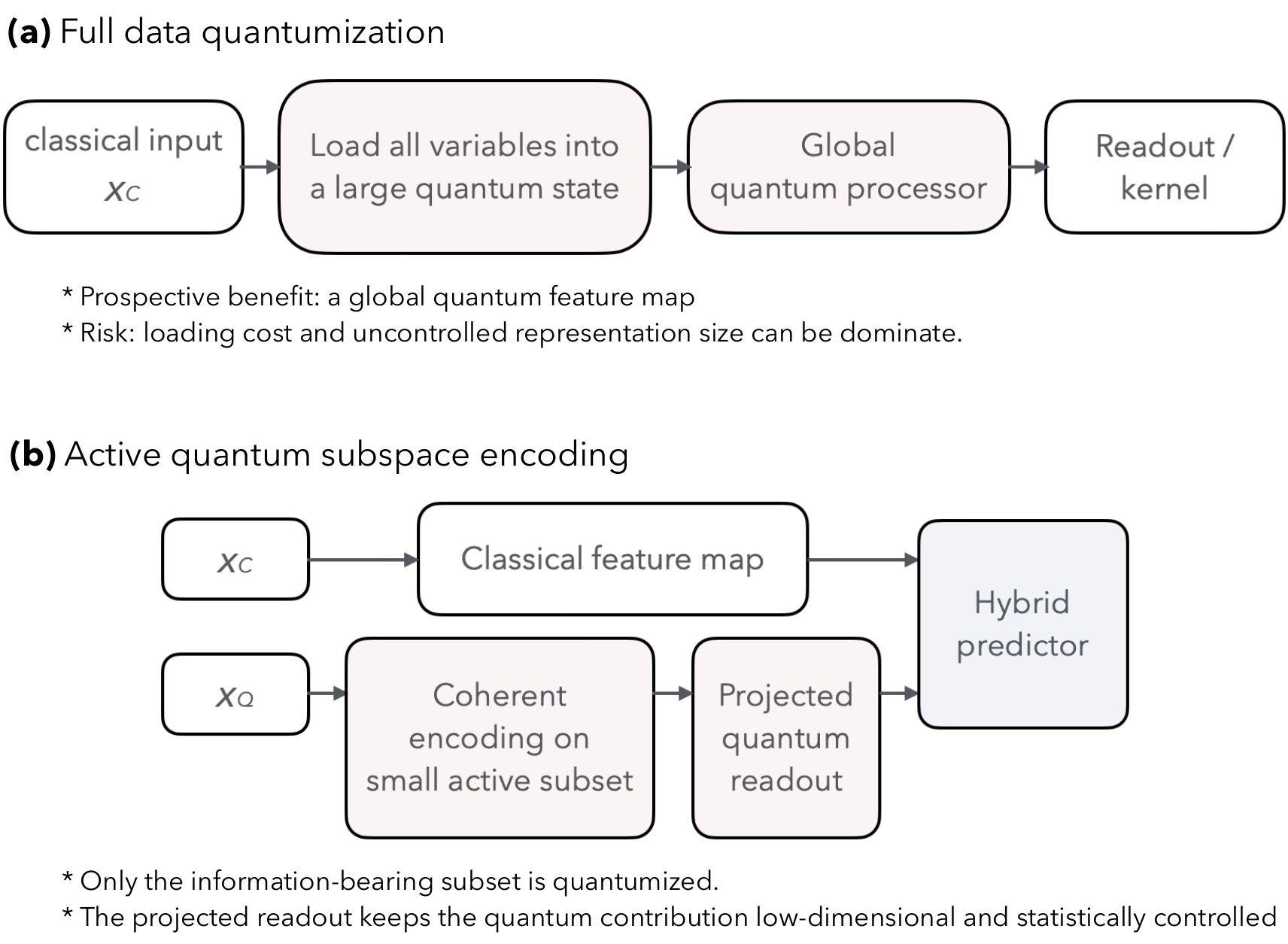}
\caption{Two views of data encoding for hybrid quantum learning. In a full quantum data-encoding approach, one attempts to load all relevant structure of the classical input into a quantum state and then process it with a global quantum model. In the active-quantum-subspace data-encoding approach proposed here, the input is split into a classical part $x_C$ and a selectively quantum-encoded part $x_Q$. Only the information-bearing subset is lifted into a coherent quantum sector, while the remaining variables stay classical. The projected readout keeps the quantum contribution low-dimensional and statistically controlled. The main analytical question of this paper is when one can simultaneously maintain polynomial encoding cost, polynomial sample regularized dimension, and inverse-polynomial oracle reliability.}
\label{fig:framework}
\end{figure*}

In this work, we consider a more structural question: when can partial quantum data-encoding yield a hybrid advantage that is useful, scalable, and statistically learnable? To address this question, we develop a general framework for partial quantum data-encoding: only a selected part of the input is lifted to a coherent quantum representation, while the rest remains classical. The key object is the \emph{active quantum subspace data-encoding} scheme. Intuitively, one need not spend quantum resources on the entire input indiscriminately. Rather, one should identify a subset of features or interactions that carry information not well captured by the classical path, quantum-encode only this part into a small coherent sector, and then read out a projected set of observables. This resource-aware viewpoint is consistent with recent generalization and model-comparison results, where the relevant quantities are not the ambient Hilbert-space dimension alone but the effective dimension, number of trainable or measured degrees of freedom, and the inductive bias of the implemented representation~\cite{Abbas2021,Caro2022,Jerbi2023}. This idea is shown schematically in Fig.~\ref{fig:framework}.

The point of view is deliberately structural. We do not claim that every partial quantum data-encoding yields an advantage. Nor do we claim a universal complexity-theoretic separation between classical and quantum learning. We identify a set of sufficient conditions under which a hybrid model can remain useful, scalable, and statistically learnable even though only part of the data is quantum-encoded. In this sense, the paper is closer in spirit to a theory of ``when the hybrid advantage is possible'' than to a blanket claim that it is generic.

The main results of the paper are as follows. First, we formalize the active quantum subspace data-encoding and projected hybrid readout. The projected quantum feature map induces a positive semidefinite kernel, and its sample regularized dimension is bounded by the number of projected observables. This gives a direct way to control the representation size. Second, we prove a necessary and sufficient criterion for actual prediction improvement under squared loss. The hybrid model improves upon the classical one if and only if the projected quantum sector contributes a direction orthogonal to the classical span and correlated with the classical residual. This result turns the heuristic phrase ``residual quantum information'' into a precise Hilbert-space statement. Third, in a realizable noisy-oracle model, we derive a PAC sample-complexity bound proportional to $\beta^{-2}$, where $\beta$ is an oracle reliability parameter defined so that the oracle agrees with the target label with probability $(1+\beta)/2$. In particular, if the projected hybrid representation has polynomially controlled dimension and the oracle reliability remains inverse-polynomial, then the learner is PAC learnable with polynomial sample complexity. This is one of the main messages of the paper: the hybrid advantage can remain statistically efficient, not merely representational. Fourth, for a canonical family based on Clifford processing, active phase encodings, and projected Pauli readout, we prove an exact dephasing-noise formula in the Heisenberg picture. This allows us to specify a scalable regime in which the oracle reliability remains inverse-polynomial although the encoding gate complexity grows polynomially with the system size, thereby putting the polynomial PAC learnability guarantee on a genuinely scalable footing.

The narrative suggested by these results is, in our view, an interesting one for the broader QML literature. The central issue is not ``how large a Hilbert space can we access?'' It is rather ``which directions in representation space actually matter, and how much quantum coherence is needed to isolate them?'' Full quantum data-encoding can certainly be useful in some regimes, but it is not necessary for quantum learning advantage. A carefully chosen active quantum subspace data-encoding is already sufficient, provided it carries information outside the classical span, yields a controlled sample regularized dimension, and preserves enough noise-robust oracle reliability to keep the resulting learner polynomially sample-efficient.

\section{Active quantum subspaces and projected hybrid readout}\label{sec:framework}

We begin by defining the representation model. The purpose of the definition is to separate three resources that are often mixed together: the total number of qubits available to the circuit, the number of qubits that are actually used to create coherent data-dependent structure at the encoding stage, and the size of the observable family used at readout.

\subsection{Definition of the representation model}

Let the input space be decomposed as
\begin{eqnarray}
\mathcal{X}_n = \mathcal{X}_{C,n}\times \mathcal{X}_{Q,n}, \quad x=(x_C,x_Q),
\end{eqnarray}
where $x_C$ denotes the variables kept on the classical path and $x_Q$ denotes the variables that are to be selectively quantum-encoded. We consider an $\kappa(n)$-qubit Hilbert space $\mathcal{H}_{\kappa(n)}=(\mathbb{C}^2)^{\otimes \kappa(n)}$ together with a partition of qubit indices
\begin{eqnarray}
[\kappa(n)] = S_n \sqcup C_n, \quad \abs{S_n}=\xi(n),
\end{eqnarray}
where $S_n$ is the \emph{active} subset and $C_n$ is the remaining subset.

\begin{definition}[Active quantum subspace data-encoding]
A family of encoders $\{\hat{E}_n\}_{n\ge 1}$ is called an active quantum subspace data-encoding (AQSE) family if, for each problem size $n$, the following hold.
\begin{enumerate}
\item[(i)] The encoder acts on $\kappa(n)$ qubits but only qubits in $S_n$ receive superposition-generating or phase-sensitive data-loading gates at the initial data-encoding step.
\item[(ii)] Qubits in $C_n$ are prepared in computational-basis states, or are classically controlled by $x_C$, at the initial data-encoding step.
\item[(iii)] The total number of elementary gates used by the encoder is
\begin{eqnarray}
G_{\mathrm{enc}}(n)=\mathrm{poly}(n).
\end{eqnarray}
\item[(iv)] The encoded data state is
\begin{eqnarray}
\hat{\rho}_x := \hat{E}_n(x)\ketbra{0^{\kappa(n)}}{0^{\kappa(n)}} \hat{E}_n(x)^\dagger.
\end{eqnarray}
\end{enumerate}
\end{definition}

The active subset $S_n$ should not be confused with the total number of qubits available to the processor. In particular, the subsequent unitary block may act on all $\kappa(n)$ qubits. What is restricted is only the coherent data-loading stage. This distinction is important. A hybrid model can use all qubits for propagation and entanglement after the encoding stage, even when only a much smaller subset is used to create the initial data-dependent superposition.

We next define the projected readout. Fix Hermitian observables
\begin{eqnarray}
\hat{O}_1,\ldots,\hat{O}_M \in \mathsf{Herm}(\mathcal{H}_{\kappa(n)}), \quad \norm{\hat{O}_a}_\infty \le 1.
\end{eqnarray}
The corresponding projected quantum feature map is
\begin{eqnarray}
\Phi_Q(x) := \big(\Tr[\hat{O}_1 \hat{\rho}_x],\ldots,\Tr[\hat{O}_M \hat{\rho}_x]\big) \in \R^M.
\label{eq:PhiQ}
\end{eqnarray}
The classical path is represented either by an explicit feature map $\Phi_C:\mathcal{X}_n\to \R^{D_C}$ or, more generally, by a classical kernel $K_C$. The hybrid predictor class is then built from
\begin{eqnarray}
\Phi_H(x) := \Phi_C(x) \oplus \sqrt{\lambda}\,\Phi_Q(x), \quad \lambda \ge 0.
\label{eq:hybrid_feature_map}
\end{eqnarray}
The associated projected hybrid kernel is
\begin{eqnarray}
K_H(x,x') = K_C(x,x') + \lambda K_Q^{\mathrm{proj}}(x,x'),
\label{eq:KH_def}
\end{eqnarray}
with
\begin{eqnarray}
K_Q^{\mathrm{proj}}(x,x') := \sum_{a=1}^{M}\Tr[\hat{O}_a \hat{\rho}_x] \Tr[\hat{O}_a \hat{\rho}_{x'}].
\label{eq:KQ_proj}
\end{eqnarray}
The adjective ``projected'' is important. We do not compare full density operators directly. Instead, we compare only a controlled family of the measured observables. This is conceptually close to the projected-kernel idea of Ref.~\cite{Huang2021}, but here it is tied explicitly to a selective encoding model. It also follows the operational lesson of classical-shadow tomography: many useful tasks require only a controlled set of expectation values rather than full state reconstruction~\cite{HuangKuengPreskill2020}.

\subsection{Sample regularized dimension is controlled by the projection size}

The first structural fact is simple but central. It says that once the quantum contribution is read out through only $M$ projected observables, the regularized dimension of the hybrid model cannot exceed the classical sample rank plus $M$.

For a finite sample $\{x_i\}_{i=1}^{N}$, let $K_H \in \R^{N \times N}$ be the Gram matrix with entries $[K_H]_{ij} = K_H(x_i, x_j)$. Define the sample regularized dimension at regularization scale $\mu > 0$ by
\begin{eqnarray}
d_{\mathrm{reg}}^{(\mu)}(K_H) := \Tr\big[K_H(K_H + \mu I_N)^{-1}\big].
\label{eq:deff_def}
\end{eqnarray}
This is a finite-sample regularized quantity; we use the term ``sample regularized dimension'' to distinguish it from the asymptotic effective-dimension notions used in statistical learning theory and QML~\cite{CaponnettoDeVito2007,Abbas2021,Caro2022}.

\begin{theorem}[Projected hybrid kernels are low-rank on samples]\label{thm:kernel_rank}
Let $K_C$ be any positive semidefinite kernel matrix on a sample of size $N$, and let $K_Q^{\mathrm{proj}}$ be defined by Eq.~(\ref{eq:KQ_proj}) using $M$ observables. Then, the hybrid Gram matrix
\begin{eqnarray}
K_H = K_C + \lambda K_Q^{\mathrm{proj}}
\end{eqnarray}
is positive semidefinite. Moreover,
\begin{eqnarray}
\rank(K_H) \le \rank(K_C) + M,
\label{eq:rank_bound}
\end{eqnarray}
and for every $\mu > 0$,
\begin{eqnarray}
d_{\mathrm{reg}}^{(\mu)}(K_H) \le \rank(K_C) + M.
\label{eq:deff_bound}
\end{eqnarray}
\end{theorem}

\begin{proof}---Define the $N \times M$ matrix $F_Q$ by
\begin{eqnarray}
[F_Q]_{ia} := \Tr[\hat{O}_a \hat{\rho}_{x_i}].
\end{eqnarray}
Then,
\begin{eqnarray}
K_Q^{\mathrm{proj}} = F_Q F_Q^\top,
\end{eqnarray}
which is positive semidefinite. Since $K_C$ is also positive semidefinite, so is
\begin{eqnarray}
K_H = K_C + \lambda F_QF_Q^\top.
\end{eqnarray}
This proves the first claim.

For the rank bound, we use the elementary inequality
\begin{eqnarray}
\rank(A+B) \le \rank(A) + \rank(B).
\end{eqnarray}
By applying it to $A=K_C$ and $B=\lambda F_Q F_Q^\top$, we have
\begin{eqnarray}
\rank(K_H) &\le& \rank(K_C) + \rank(F_Q F_Q^\top) \nonumber \\
	&\le& \rank(K_C) + M,
\end{eqnarray}
because $\rank(F_Q F_Q^\top) \le \rank(F_Q) \le M$. This proves Eq.~(\ref{eq:rank_bound}).

Finally, let $\sigma_1,\ldots,\sigma_r>0$ be the nonzero eigenvalues of $K_H$, where $r=\rank(K_H)$. By definition,
\begin{eqnarray}
d_{\mathrm{reg}}^{(\mu)}(K_H) = \sum_{i=1}^{r}\frac{\sigma_i}{\sigma_i+\mu}.
\end{eqnarray}
Each term is at most $1$, hence
\begin{eqnarray}
d_{\mathrm{reg}}^{(\mu)}(K_H) \le r = \rank(K_H) \le \rank(K_C) + M.
\end{eqnarray}
This proves Eq.~(\ref{eq:deff_bound}).
\end{proof}

{\bf Theorem~\ref{thm:kernel_rank}} already captures one important difference between the present framework and naive fidelity-type kernels. Here, the quantum contribution is measured through a finite projected family, so the sample regularized dimension is directly controlled. In particular, if $M=\mathrm{poly}(n)$ and the classical sample rank is also polynomial, then the hybrid feature space remains statistically manageable even when the ambient Hilbert space dimension is exponentially large.

The theorem is intentionally modest. It does not say that small sample regularized dimension is sufficient for good learning. It says only that a projected hybrid readout prevents the uncontrolled rank explosion that often appears in global overlap kernels. This point is complementary to recent concentration results for global quantum kernels, where excessive expressivity, global measurements, entanglement, or noise can wash out data dependence unless the representation is carefully controlled~\cite{Thanasilp2024,Agliardi2026}. This is the first of the three ingredients we need: a representation that is not too large.

\section{When does partial quantum data-encoding actually help?}\label{sec:benefit}

A low-dimensional projected quantum readout is useful only if it adds directions that the classical path does not already contain. This section makes this statement precise. The result below turns the informal phrase ``the quantum sector contains residual information'' into a complete Hilbert-space criterion.

Let $(X,Y)$ be a random pair with $Y \in L_2$. Write $\langle f,g \rangle := \E[f(X) g(X)]$ and $\norm{f}_2 := \sqrt{\langle f, f \rangle}$ for square-integrable functions of $X$. Let $\mathcal{V}_C \subset L_2(P_X)$ be the closed linear span of the classical features and $\mathcal{V}_Q \subset L_2(P_X)$ the finite-dimensional span of the projected quantum features in Eq.~(\ref{eq:PhiQ}). Define
\begin{eqnarray}
\mathcal{V}_H := \mathcal{V}_C + \mathcal{V}_Q.
\end{eqnarray}
Let $P_C$ and $P_H$ denote the orthogonal projections onto $\mathcal{V}_C$ and $\mathcal{V}_H$, respectively. The optimal squared risks in the classical and hybrid spaces are
\begin{eqnarray}
R_C^* := \E[(Y-P_CY)^2], \quad R_H^* := \E[(Y-P_HY)^2].
\end{eqnarray}
We also define the classical residual
\begin{eqnarray}
r_C := Y-P_CY.
\label{eq:classical_residual}
\end{eqnarray}

\begin{theorem}[Residual-information criterion for hybrid benefit]\label{thm:residual}
With the notation above, the following statements hold.
\begin{enumerate}
\item[(i)] One always has $R_H^*\le R_C^*$.
\item[(ii)] For any $u\in \mathcal V_Q$, let
\begin{eqnarray}
u_\perp := (I - P_C) u.
\end{eqnarray}
If $u_\perp \neq 0$, then
\begin{eqnarray}
R_C^* - R_H^* \ge \frac{\langle r_C, u_\perp \rangle^2}{\norm{u_\perp}_2^2}.
\label{eq:residual_quantitative}
\end{eqnarray}
\item[(iii)] One has strict improvement, $R_H^* < R_C^*$, if and only if there exists some $u \in \mathcal{V}_Q$ such that $u_\perp \neq 0$ and
\begin{eqnarray}
\langle r_C, u_\perp \rangle \neq 0.
\label{eq:residual_if_and_only_if}
\end{eqnarray}
\end{enumerate}
\end{theorem}

\begin{proof}---Since $\mathcal{V}_C \subseteq \mathcal{V}_H$, orthogonal projection onto the larger subspace cannot increase the squared error. Hence, $R_H^* \le R_C^*$, proving part (i).

For part (ii), fix $u \in \mathcal{V}_Q$ and define the one-dimensional augmentation
\begin{eqnarray}
\mathcal{V}_u := \mathcal{V}_C + \spanop\{u\}.
\end{eqnarray}
Because $\mathcal{V}_u \subseteq \mathcal{V}_H$, the optimal risk in $\mathcal{V}_H$ is no larger than the optimal risk in $\mathcal{V}_u$. Therefore, it suffices to compute the gain obtained by adding the single direction $u$.

If $u_\perp = (I - P_C)u=0$, then $u \in \mathcal{V}_C$ and no gain is possible. So, assume $u_\perp \neq 0$. By construction, $u_\perp \perp \mathcal{V}_C$. Therefore, the orthogonal projection of $Y$ onto $\mathcal{V}_u$ is
\begin{eqnarray}
P_{\mathcal{V}_u} Y = P_C Y + \frac{\langle r_C, u_\perp \rangle}{\norm{u_\perp}_2^2} u_\perp,
\end{eqnarray}
because $r_C=Y - P_CY$ is orthogonal to $\mathcal{V}_C$. The corresponding squared error is
\begin{eqnarray}
\E\big[(Y-P_{\mathcal{V}_u} Y)^2\big] &=& \norm{r_C - \frac{\langle r_C, u_\perp\rangle}{\norm{u_\perp}_2^2} u_\perp}_2^2 \\
	&=& \norm{r_C}_2^2 - \frac{\langle r_C, u_\perp \rangle^2}{\norm{u_\perp}_2^2}
\end{eqnarray}
by the Pythagorean theorem. Since $\norm{r_C}_2^2 = R_C^*$ and $R_H^* \le \E[(Y - P_{\mathcal{V}_u}Y)^2]$, we obtain Eq.~(\ref{eq:residual_quantitative}). This proves part (ii).

For part (iii), first suppose there exists $u \in \mathcal{V}_Q$ with $u_\perp \neq 0$ and $\langle r_C, u_\perp\rangle \neq 0$. Then, Eq.~(\ref{eq:residual_quantitative}) implies $R_H^* < R_C^*$. Conversely, suppose no such $u$ exists. Then, for every $u \in \mathcal{V}_Q$ one has $\langle r_C, (I-P_C)u \rangle=0$. Since $r_C \perp \mathcal{V}_C$, this implies $\langle  r_C, u \rangle=0$ for all $u \in \mathcal{V}_Q$. Hence, $r_C \perp \mathcal{V}_Q$ and also $r_C \perp \mathcal{V}_C$, so $r_C \perp \mathcal{V}_H = \mathcal{V}_C + \mathcal{V}_Q$. Therefore, $P_H Y = P_C Y$ and $R_H^* = R_C^*$. This proves the converse implication.
\end{proof}

{\bf Theorem~\ref{thm:residual}} is one of the conceptual cores of this work. It says that the hybrid model is useful precisely when the projected quantum sector contains a direction that survives orthogonalization against the classical span and still correlates with what the classical model fails to explain. This is the exact sense in which partial quantum data-encoding can be enough. One does not need a huge quantum state space. One only needs at least one useful residual direction.

The theorem also clarifies the relation to the geometric-difference narrative in Ref.~\cite{Huang2021}. In that work, a quantum model becomes interesting only when its induced geometry differs from that of the classical competitor. Here, the same point appears in function-space language. The component $u_\perp=(I-P_C)u$ is precisely the part of a projected quantum direction that is geometrically invisible to the classical span. If $u_\perp=0$ for every projected quantum feature, then the hybrid model adds no new geometry at all. Appendix~\ref{app:residual_exact} strengthens this statement by giving an exact projection formula for the gain and a finite-sample matrix criterion that can be used in practice to diagnose whether a candidate active subspace is useful.

{\bf Theorem~\ref{thm:residual}} is stated for squared loss because that setting allows a complete and transparent orthogonal-projection proof. The message, however, is broader. In classification, one usually works with a surrogate risk or a margin condition rather than a literal $L_2$ projection. In that setting, the exact orthogonal-projection identity is replaced by the corresponding excess-risk or margin error left after fitting the classical predictor, so the same design lesson remains: the hybrid sector is useful only when it contributes directions that are not already present on the classical path and that align with the remaining prediction error.

\section{Noisy hybrid learning and PAC sample complexity}\label{sec:noisy}

We turn to the statistical question. Suppose the learner receives labels through a noisy oracle. How does the hybrid oracle reliability enter the sample complexity? This is the point where the present framework meets the sample-complexity mechanism of Ref.~\cite{Song2021}. The sample-complexity viewpoint is also complementary to recent quantum-learning results where advantage is formulated as a reduction in the number of experiments, measurements, or training samples rather than only as a time-to-solution speedup~\cite{Huang2022Science,Caro2022,Liu2025Science,King2024}.

We work in a realizable binary setting in the sense of PAC learning~\cite{Valiant1984}, together with an input-dependent classification-noise model~\cite{AngluinLaird1988}. Let $\mathcal H$ be a hypothesis class of functions from $\mathcal{X}_n$ to $\{-1,+1\}$. We assume there exists a target hypothesis $h^\star \in \mathcal{H}$. The learner does not observe $h^\star(X)$ directly. Instead, for each input $x$, the oracle returns a noisy label $\widetilde{Y} \in \{-1,+1\}$ satisfying
\begin{eqnarray}
\Prob\big(\widetilde Y = h^\star(x)\mid X=x\big) = \frac{1+\beta(x)}{2},
\label{eq:noisy_model}
\end{eqnarray}
where $0< \beta_0 \le \beta(x) \le 1$ and
\begin{eqnarray}
\beta_0 := \inf_{x\in \mathcal{X}_n}\beta(x)
\label{eq:beta0_def}
\end{eqnarray}
is the worst-case oracle reliability. Thus larger $\beta(x)$ means a more reliable oracle: $\beta(x)=1$ is noiseless, while $\beta(x)=0$ is maximally uninformative. The clean risk and noisy risk are
\begin{eqnarray}
R(h) &:=& \Prob\big(h(X)\neq h^\star(X)\big), \nonumber \\
R_\eta(h) &:=& \Prob\big(h(X)\neq \widetilde Y\big).
\end{eqnarray}
Given $N$ i.i.d. samples $\{(X_i,\widetilde Y_i)\}_{i=1}^{N}$, define the empirical noisy risk
\begin{eqnarray}
\widehat{R}_\eta(h) := \frac{1}{N}\sum_{i=1}^{N}\mathds{I}\{h(X_i) \neq \widetilde{Y}_i\}.
\end{eqnarray}
Let $\widehat{h}$ be any empirical risk minimizer over $\mathcal{H}$.

The next theorem shows that the sample complexity scales as $\beta_0^{-2}$. This generalizes the inverse-square dependence on oracle reliability in Ref.~\cite{Song2021} from a specific oracle construction to an arbitrary realizable hybrid hypothesis class.

\begin{theorem}[PAC learning with a noisy hybrid oracle]\label{thm:noisy_pac}
Let $\mathcal{H}$ be a binary hypothesis class with VC dimension $d < \infty$, and suppose the noisy oracle model as in Eq.~(\ref{eq:noisy_model}) holds with the reliability lower bound $\beta_0>0$. Let $\widehat{h}$ be an empirical minimizer of $\widehat{R}_\eta$ over $\mathcal{H}$. Then, the following hold.
\begin{enumerate}
\item[(i)] For every $h \in \mathcal{H}$,
\begin{eqnarray}
R_\eta(h) - R_\eta(h^\star) \ge \beta_0 R(h).
\label{eq:noisy_clean_relation}
\end{eqnarray}
\item[(ii)] On every sample on which
\begin{eqnarray}
\sup_{h \in \mathcal{H}}\abs{\widehat{R}_\eta(h) - R_\eta(h)} \le \alpha,
\label{eq:unif_alpha}
\end{eqnarray}
one has
\begin{eqnarray}
R(\widehat{h}) \le \frac{2\alpha}{\beta_0}.
\label{eq:Rhat_bound}
\end{eqnarray}
\item[(iii)] In particular, if
\begin{eqnarray}
N \ge \frac{32}{\beta_0^2\varepsilon^2} \left[ d\log\left(\frac{2eN}{d}\right)+\log\left(\frac{8}{\delta}\right)\right],
\label{eq:sample_bound_exact}
\end{eqnarray}
then with probability at least $1-\delta$,
\begin{eqnarray}
R(\widehat{h}) \le \varepsilon.
\label{eq:pac_result}
\end{eqnarray}
Thus the sample complexity is of order
\begin{eqnarray}
N = O\left(\frac{d\log(1/\varepsilon)+\log(1/\delta)}{\beta_0^2\varepsilon^2}\right).
\label{eq:sample_bound_asym}
\end{eqnarray}
\end{enumerate}
\end{theorem}

\begin{proof}---For part (i), define the input-dependent label-flip probability
\begin{eqnarray}
\eta(x) := \Prob(\widetilde{Y} \neq h^\star(x) \mid X=x)=\frac{1-\beta(x)}{2}.
\end{eqnarray}
Conditioned on $X=x$, the hypothesis $h$ makes a noisy error either because the oracle flips the correct label, or because $h$ disagrees with $h^\star$ and the oracle does not flip the label. Therefore, 
\begin{eqnarray}
&& \Prob\big(h(X) \neq \widetilde{Y} \mid X=x\big) \nonumber \\
&& \quad = \eta(x) + (1-2\eta(x))\mathds{I}\{h(x) \neq h^\star(x)\}.
\end{eqnarray}
Taking expectation over $X$ gives
\begin{eqnarray}
R_\eta(h) = \E[\eta(X)] + \E[\beta(X)\mathds{1}\{h(X)\neq h^\star(X)\}].
\end{eqnarray}
Since $R_\eta(h^\star)=\E[\eta(X)]$, we obtain
\begin{eqnarray}
R_\eta(h)-R_\eta(h^\star) &=& \E[\beta(X)\mathds{1}\{h(X)\neq h^\star(X)\}] \nonumber \\
	&\ge& \beta_0 R(h),
\end{eqnarray}
which proves Eq.~(\ref{eq:noisy_clean_relation}).

For part (ii), note first that because $\widehat{h}$ minimizes $\widehat{R}_\eta$ over $\mathcal{H}$,
\begin{eqnarray}
\widehat{R}_\eta(\widehat{h}) \le \widehat{R}_\eta(h^\star).
\end{eqnarray}
Assuming Eq.~(\ref{eq:unif_alpha}), we have
\begin{eqnarray}
R_\eta(\widehat{h}) - R_\eta(h^\star) &\le& \big(R_\eta(\widehat{h})-\widehat{R}_\eta(\widehat{h})\big) + \big(\widehat{R}_\eta(h^\star) - R_\eta(h^\star)\big) \nonumber \\
	&\le& 2\alpha.
\end{eqnarray}
Combining this with part (i) gives
\begin{eqnarray}
\beta_0 R(\widehat{h}) \le R_\eta(\widehat{h})-R_\eta(h^\star) \le 2\alpha,
\end{eqnarray}
which proves Eq.~(\ref{eq:Rhat_bound}).

For part (iii), we use the standard VC uniform convergence bound for binary losses:
\begin{widetext}
\begin{eqnarray}
\Prob\left(\sup_{h\in\mathcal H}\abs{\widehat R_\eta(h)-R_\eta(h)} > \alpha\right) \le 4\left(\frac{2eN}{d}\right)^d e^{-N\alpha^2/8}.
\label{eq:VC_bound}
\end{eqnarray}
\end{widetext}
Setting $\alpha=\beta_0\varepsilon/2$, the right-hand side of Eq.~(\ref{eq:VC_bound}) is at most $\delta$ whenever Eq.~(\ref{eq:sample_bound_exact}) holds. Under this event, Eq.~(\ref{eq:Rhat_bound}) implies $R(\widehat{h}) \le \varepsilon$. This proves Eq.~(\ref{eq:pac_result}). The asymptotic form in Eq.~(\ref{eq:sample_bound_asym}) is immediate.
\end{proof}

The theorem gives a clean operational meaning to oracle reliability. If the quantum part of the hybrid model increases the reliability of the queried labels, then the sample complexity is reduced by a quadratic factor in that reliability. This recovers the mechanism of Ref.~\cite{Song2021} in a much broader hypothesis-class setting. For completeness, Appendix~\ref{app:VC} includes a self-contained proof of the VC-type uniform convergence inequality used in Eq.~(\ref{eq:VC_bound}).

\begin{corollary}[Polynomial sample complexity for projected hybrid models]\label{cor:poly_sample}
Suppose the hybrid learner uses linear separators in a feature space of dimension
\begin{eqnarray}
D_H(n) := \rank(K_C) + M,
\end{eqnarray}
where $M$ is the number of projected quantum observables on the sample. If
\begin{eqnarray}
D_H(n)=\mathrm{poly}(n) \quad\text{and}\quad \beta_0(n) \ge \frac{1}{\mathrm{poly}(n)},
\label{eq:poly_conditions_basic}
\end{eqnarray}
then the noisy hybrid learner is PAC learnable with polynomial sample complexity in $n$, $1/\varepsilon$, and $\log(1/\delta)$.
\end{corollary}

\begin{proof}---The VC dimension of linear separators in $D_H$ dimensions is at most $D_H+1$. Hence, {\bf Theorem~\ref{thm:noisy_pac}} implies a sample bound of order
\begin{eqnarray}
N = O\left(\frac{D_H(n)\log(1/\varepsilon)+\log(1/\delta)}{\beta_0(n)^2\varepsilon^2}\right).
\end{eqnarray}
Under the polynomial bounds in Eq.~(\ref{eq:poly_conditions_basic}), the right-hand side is polynomial.
\end{proof}

\begin{corollary}[Relative gain from improved oracle reliability]\label{cor:relative_gain}
Consider the same realizable hypothesis class $\mathcal{H}$ queried through two oracles, a classical one with reliability $\beta_C$ and a hybrid one with reliability $\beta_H$, where $0 < \beta_C \le \beta_H \le 1$. Then, the upper bound of {\bf Theorem~\ref{thm:noisy_pac}} implies
\begin{eqnarray}
\frac{N_H(\varepsilon,\delta)}{N_C(\varepsilon,\delta)} \lesssim \left(\frac{\beta_C}{\beta_H}\right)^2
\end{eqnarray}
up to the common VC factor $d\log(2eN/d)+\log(8/\delta)$ from Eq.~(\ref{eq:sample_bound_exact}). Hence, any persistent improvement in oracle reliability translates directly into a quadratic reduction of sample complexity.
\end{corollary}

\section{Polynomial persistence under local noise}\label{sec:scalable}

We now address the main scalability question of this work. Can a hybrid advantage survive when the number of qubits and the encoding circuit size both grow, provided the growth is polynomial? To answer this, we give a canonical construction for which the effect of local noise can be tracked exactly. The model is restricted but analytically clean: active phase encoding, Clifford processing, projected Pauli readout, and local dephasing noise. This focus on projected or local observables is consistent with the broader trainability literature, where global cost functions and noise can suppress useful learning signals, while local or problem-adapted readouts can preserve trainable scales~\cite{McClean2018,Cerezo2021BP,Wang2021Noise,Thanasilp2024}.

\subsection{Canonical active-subspace family}

Fix an active subset $S_n \subset [\kappa(n)]$ and a context subset $C_n=[\kappa(n)] \setminus S_n$. For an input $x=(b,\phi)$ with
\begin{eqnarray}
b \in \{0,1\}^{\abs{C_n}}, \quad \phi=(\phi_j)_{j \in S_n} \in \R^{\abs{S_n}},
\end{eqnarray}
consider the encoded product state
\begin{eqnarray}
\ket{\psi_x} := \bigotimes_{j \in S_n} \hat{R}_z(\phi_j)\hat{H}\ket{0} \otimes \ket{b}_{C_n}.
\label{eq:canonical_state}
\end{eqnarray}
Let $\hat{V}_n$ be any Clifford circuit acting on all $\kappa(n)$ qubits, and let $\hat{P}_n$ be a Pauli readout observable. By ``score'' we mean the noiseless expectation value of the projected Pauli readout associated with the input $x$. The ideal score is
\begin{eqnarray}
s_n(x) := \bra{\psi_x} \hat{V}_n^\dagger \hat{P}_n \hat{V}_n \ket{\psi_x}.
\label{eq:ideal_score}
\end{eqnarray}
Since $\hat{V}_n$ is Clifford, the Heisenberg-evolved observable
\begin{eqnarray}
\hat{Q}_n := \hat{V}_n^\dagger \hat{P}_n \hat{V}_n
\end{eqnarray}
is again a Pauli string.

The next proposition gives the exact ideal score formula whenever the readout is compatible with the active split.

\begin{proposition}[Exact ideal score for the canonical family]\label{prop:ideal_score}
Assume that the Pauli string $\hat{Q}_n$ has only $\hat{X}$ or $\hat{Y}$ on active qubits and only $\hat{\id}$ or $\hat{Z}$ on the context qubits, namely
\begin{eqnarray}
\hat{Q}_n = \xi_n \left[\bigotimes_{j \in A_X} \hat{X}_j\right] \left[\bigotimes_{j\in A_Y} \hat{Y}_j\right] \left[\bigotimes_{\ell\in B_Z} \hat{Z}_\ell\right],
\label{eq:Qn_form}
\end{eqnarray}
where $\xi_n \in \{\pm 1\}$, $A_X, A_Y\subseteq S_n$, and $B_Z \subseteq C_n$. Then, the ideal score in Eq.~(\ref{eq:ideal_score}) factorizes as
\begin{eqnarray}
s_n(x) = \xi_n (-1)^{\sum_{\ell \in B_Z} b_\ell} \prod_{j \in A_X}\cos\phi_j \prod_{j \in A_Y}\sin\phi_j.
\label{eq:ideal_score_formula}
\end{eqnarray}
If $\hat{Q}_n$ contains any $\hat{Z}$ acting on an active qubit or any $\hat{X}$/$\hat{Y}$ acting on a context qubit, then $s_n(x)=0$ for all $x$.
\end{proposition}

\begin{proof}---The state in Eq.~(\ref{eq:canonical_state}) is a product state across qubits, so the expectation of a tensor-product observable factorizes into the product of one-qubit expectations. For an active qubit,
\begin{eqnarray}
\ket{+_{\phi}} := \hat{R}_z(\phi)\hat{H}\ket{0} = \frac{\ket{0} + e^{i\phi}\ket{1}}{\sqrt{2}},
\end{eqnarray}
which satisfies
\begin{eqnarray}
\bra{+_{\phi}}\hat{X}\ket{+_{\phi}} &=& \cos\phi, \nonumber \\
\bra{+_{\phi}}\hat{Y}\ket{+_{\phi}} &=& \sin\phi, \nonumber \\
\bra{+_{\phi}}\hat{Z}\ket{+_{\phi}} &=& 0.
\label{eq:single_qubit_expt}
\end{eqnarray}
For a context qubit prepared in a basis state $\ket{b_\ell}$, one has
\begin{eqnarray}
\bra{b_\ell}\hat{Z}\ket{b_\ell}=(-1)^{b_\ell}, ~~ \bra{b_\ell}\hat{X}\ket{b_\ell}=\bra{b_\ell}\hat{Y}\ket{b_\ell}=0.
\label{eq:context_expt}
\end{eqnarray}
Therefore, if $\hat{Q}_n$ has the form in Eq.~(\ref{eq:Qn_form}), the expectation factorizes exactly into the product of the corresponding terms in Eq.~(\ref{eq:single_qubit_expt}) and Eq.~(\ref{eq:context_expt}), yielding Eq.~(\ref{eq:ideal_score_formula}). If instead $\hat{Q}_n$ contains a $\hat{Z}$ on an active qubit or an $\hat{X}$/$\hat{Y}$ on a context qubit, one of the one-qubit factors is zero, so the whole product vanishes.
\end{proof}

{\bf Proposition~\ref{prop:ideal_score}} is already informative. A projected Pauli observable can compress a high-order multiplicative interaction of the phases and context bits into a single scalar quantity. This is exactly the type of situation in which a small active quantum sector can add a direction outside the classical low-order span.

\subsection{Exact dephasing-noise scaling}

We now place the canonical family under local dephasing noise. After each ideal Clifford gate $\hat{U}_g$ in the circuit $\hat{V}_n = \hat{U}_G \cdots \hat{U}_1$, suppose each qubit in the device is followed by a dephasing channel
\begin{eqnarray}
\mathcal Z_{p_g}(\hat{\rho})=(1-p_g)\hat{\rho} + p_g \hat{Z}\hat{\rho} \hat{Z}, \quad 0 \le p_g \le \frac{1}{2}.
\end{eqnarray}
In the Heisenberg picture,
\begin{eqnarray}
\mathcal Z_{p_g}^\ast(\hat{I}) &=& \hat{I}, \nonumber \\
\mathcal Z_{p_g}^\ast(\hat{Z}) &=& \hat{Z}, \nonumber \\
\mathcal Z_{p_g}^\ast(\hat{X}) &=& (1-2p_g)\hat{X}, \nonumber \\
\mathcal Z_{p_g}^\ast(\hat{Y}) &=& (1-2p_g)\hat{Y}.
\label{eq:dephasing_heisenberg}
\end{eqnarray}
Because Clifford conjugation maps Pauli strings to Pauli strings, the backward propagation of $\hat{P}_n$ through the noisy circuit can be tracked exactly, location by location.

Let $L_n(\hat{P}_n)$ denote the set of noise locations in the backward light cone of $\hat{P}_n$ at which the relevant one-qubit Pauli operator is $\hat{X}$ or $\hat{Y}$. Equivalently, these are precisely the places where the factor $(1-2p_g)$ is picked up under Eq.~(\ref{eq:dephasing_heisenberg}).

\begin{theorem}[Exact score attenuation for Clifford AQSE under dephasing]\label{thm:dephasing}
Consider the canonical active-subspace family above. Let $\hat{\rho}_x^{\mathrm{noisy}}$ denote the output state when the ideal Clifford circuit $\hat{V}_n$ is interleaved with local dephasing channels as described above. Then for every input $x$,
\begin{eqnarray}
\Tr\big[\hat{P}_n \hat{\rho}_x^{\mathrm{noisy}}\big] = \Lambda_n(\hat{P}_n) s_n(x),
\label{eq:noisy_exact}
\end{eqnarray}
where $s_n(x)$ is the ideal score from Eq.~(\ref{eq:ideal_score}) and $\Lambda_n(\hat{P}_n)$ is the dephasing attenuation factor
\begin{eqnarray}
\Lambda_n(\hat{P}_n) := \prod_{g \in L_n(\hat{P}_n)} (1-2p_g).
\label{eq:Lambda_def}
\end{eqnarray}
In particular, if $p_g \le \kappa/\abs{L_n(\hat{P}_n)}$ for all relevant noise locations and $\abs{L_n(\hat{P}_n)} \ge 4\kappa$, then
\begin{eqnarray}
\Lambda_n(\hat{P}_n) \ge e^{-4\kappa}.
\label{eq:Lambda_lower}
\end{eqnarray}
More generally, if $p_g=O\left(\frac{\log n}{\abs{L_n(\hat{P}_n)}}\right)$, then $\Lambda_n(\hat{P}_n) \ge n^{-O(1)}$.
\end{theorem}

\begin{proof}---Write the noisy channel in the Heisenberg picture as
\begin{eqnarray}
\mathcal{C}_n^\ast = \mathcal{U}_1^\ast \circ \mathcal{Z}_{p_1}^\ast \circ \cdots \circ \mathcal{U}_G^\ast \circ \mathcal{Z}_{p_G}^\ast,
\end{eqnarray}
where $\mathcal U_g^\ast(\hat{A})=\hat{U}_g^\dagger \hat{A} \hat{U}_g$. Since each $\hat{U}_g$ is Clifford, if we start from the Pauli observable $\hat{P}_n$, then after each application of $\mathcal{U}_g^\ast$, we still have a Pauli string. At the subsequent dephasing step, every one-qubit factor equal to $\hat{\id}$ or $\hat{Z}$ is unchanged, while every factor equal to $\hat{X}$ or $\hat{Y}$ is multiplied by $(1-2p_g)$ by Eq.~(\ref{eq:dephasing_heisenberg}). Therefore, after all steps,
\begin{eqnarray}
\mathcal{C}_n^\ast(\hat{P}_n) = \left[\prod_{g \in L_n(\hat{P}_n)}\!\!\!\!\!(1-2p_g)\right]\hat{Q}_n = \Lambda_n(\hat{P}_n) \hat{Q}_n,
\end{eqnarray}
where $\hat{Q}_n = \hat{V}_n^\dagger \hat{P}_n \hat{V}_n$ is the ideal Heisenberg image and $L_n(\hat{P}_n)$ is exactly the set of noise locations where an $\hat{X}$ or $\hat{Y}$ factor occurs. Taking expectation in the encoded state $\ket{\psi_x}$ gives
\begin{eqnarray}
\Tr\big[\hat{P}_n \hat{\rho}_x^{\mathrm{noisy}}\big] &=& \bra{\psi_x}\mathcal C_n^\ast(\hat{P}_n)\ket{\psi_x} \nonumber \\
	&=& \Lambda_n(\hat{P}_n)\bra{\psi_x}\hat{Q}_n\ket{\psi_x},
\end{eqnarray}
which is Eq.~(\ref{eq:noisy_exact}) because $\bra{\psi_x}\hat{Q}_n\ket{\psi_x}=s_n(x)$.

To prove Eq.~(\ref{eq:Lambda_lower}), let $L := \abs{L_n(\hat{P}_n)}$. By assumption $p_g \le \kappa/L$ for all relevant $g$, so
\begin{eqnarray}
\Lambda_n(\hat{P}_n) \ge \left(1-\frac{2\kappa}{L}\right)^L.
\end{eqnarray}
If $L \ge 4\kappa$, then $0 \le 2\kappa/L \le 1/2$, and the elementary bound $\log(1-t) \ge -2t$ for $0 \le t \le 1/2$ gives
\begin{eqnarray}
L\log\left(1-\frac{2\kappa}{L}\right) \ge -4\kappa.
\end{eqnarray}
Exponentiating yields Eq.~(\ref{eq:Lambda_lower}). The final statement follows similarly by replacing $\kappa/L$ with $c(\log n)/L$.
\end{proof}

{\bf Theorem~\ref{thm:dephasing}} gives a precise form to the statement that polynomial gate complexity does not by itself destroy the relevant score scale. What matters is not just how many gates exist in total, but how the relevant observable propagates through the backward light cone and how much attenuation accumulates there. The factor $\Lambda_n(\hat{P}_n)$ therefore attenuates the readout score itself; the corresponding statement about oracle reliability is derived next. Appendix~\ref{app:noise_general} extends the same calculation to general local Pauli-diagonal noise, and Appendix~\ref{app:explicit_family} packages the result into an explicit scalable family with logarithmic active support.

\begin{corollary}[Polynomial persistence of oracle reliability]\label{cor:poly_persistence}
Assume the setting of {\bf Theorem~\ref{thm:dephasing}} and {\bf Proposition~\ref{prop:ideal_score}}. Suppose further that on the data support the ideal score satisfies
\begin{eqnarray}
\abs{s_n(x)} \ge s_{\min}^{\mathrm{id}}(n) \quad\text{for all $x$},
\label{eq:ideal_bias_bound}
\end{eqnarray}
where $s_{\min}^{\mathrm{id}}(n) \ge 1/\mathrm{poly}(n)$ is a uniform lower bound on the noiseless score. Define the target label by
\begin{eqnarray}
h^\star(x)=\sgn(s_n(x)).
\end{eqnarray}
If the learner queries the noisy Pauli measurement of $\hat{P}_n$, then the resulting oracle reliability obeys
\begin{eqnarray}
\beta_0(n)\ge \Lambda_n(\hat{P}_n)\,s_{\min}^{\mathrm{id}}(n).
\label{eq:beta_total_bound}
\end{eqnarray}
Hence, whenever $\Lambda_n(\hat{P}_n)\ge 1/\mathrm{poly}(n)$, the oracle reliability remains inverse-polynomial. Combined with {\bf Corollary~\ref{cor:poly_sample}}, this implies polynomial PAC sample complexity even when the number of qubits and the encoding gate count both scale polynomially with $n$.
\end{corollary}

\begin{proof}---The outcome of the Pauli measurement is a random variable $M(x) \in \{-1,+1\}$ with mean
\begin{eqnarray}
\E[M(x)\mid X=x] = \Tr[\hat{P}_n\hat{\rho}_x^{\mathrm{noisy}}].
\end{eqnarray}
Because the target label is $h^\star(x)=\sgn(s_n(x))$, the correctness probability of the queried label is
\begin{eqnarray}
\Prob\big(M(x)=h^\star(x)\mid X=x\big) &=& \frac{1+h^\star(x)\E[M(x) \mid X=x]}{2} \nonumber \\
	&=& \frac{1+\abs{\Tr[\hat{P}_n\hat{\rho}_x^{\mathrm{noisy}}]}}{2}.
\end{eqnarray}
Therefore, the oracle reliability is
\begin{eqnarray}
\beta(x)=\abs{\Tr[\hat{P}_n\hat{\rho}_x^{\mathrm{noisy}}]}.
\end{eqnarray}
By {\bf Theorem~\ref{thm:dephasing}} and Eq.~(\ref{eq:ideal_bias_bound}),
\begin{eqnarray}
\beta(x) = \Lambda_n(\hat{P}_n)\abs{s_n(x)} \ge \Lambda_n(\hat{P}_n)s_{\min}^{\mathrm{id}}(n).
\end{eqnarray}
Taking the infimum over $x$ proves Eq.~(\ref{eq:beta_total_bound}). The final statement follows from {\bf Corollary~\ref{cor:poly_sample}}.
\end{proof}

{\bf Corollary~\ref{cor:poly_persistence}} is the precise form of the main scalability message. If the ideal active-subspace score floor is at least inverse-polynomial and the noise attenuation along the relevant light cone is also inverse-polynomial, then polynomial encoding cost does not destroy the PAC learnability. In particular, if the attenuation stays constant, then the only asymptotic decay comes from the intentionally chosen ideal score floor itself.

A concrete and useful regime is the following. Suppose the score in Eq.~(\ref{eq:ideal_score_formula}) uses only $a(n)=O(\log n)$ active qubits and the phase range is restricted so that $\abs{\cos\phi_j} \ge c_0 > 0$ or $\abs{\sin\phi_j} \ge c_0 > 0$ on the data support. Then,
\begin{eqnarray}
s_{\min}^{\mathrm{id}}(n)\ge c_0^{a(n)} = n^{-O(1)}.
\end{eqnarray}
If, in addition, each relevant dephasing probability satisfies $p_g\le \kappa/|L_n|$, then {\bf Theorem~\ref{thm:dephasing}} gives a constant lower bound on $\Lambda_n(\hat{P}_n)$. This is a direct and explicit realization of an active-subspace regime in which the number of qubits and the encoding gate count may both grow polynomially while the oracle reliability remains inverse-polynomial.

\section{\texorpdfstring{Large-qubit active-subspace toy family and benchmark numerics}{Large-qubit active-subspace toy family and benchmark numerics}}\label{sec:toy}

Here, we give an example of a $64$-qubit toy family whose active support remains small while the total register is genuinely large. The point is not to brute-force a $2^{64}$-dimensional statevector, but to keep the projected quantum feature analytically tractable even when the ambient device size is far beyond the original toy scale.

\subsection{\texorpdfstring{A 64-qubit toy family with six active phases}{A 64-qubit toy family with six active phases}}

Let the active subset be
\begin{eqnarray}
S=\{1,2,3,4,5,6\},
\end{eqnarray}
and let the remaining qubits
\begin{eqnarray}
C=\{7,8,\ldots,64\}
\end{eqnarray}
stay on the classical path at the data-loading stage. For an input $x=(b,\phi)$ with $b \in \{0,1\}^{58}$ and $\phi=(\phi_1,\ldots,\phi_6) \in [0,2\pi)^6$, consider the encoded state
\begin{eqnarray}
\ket{\Psi_x} := \hat{V}_{64} \left( \bigotimes_{j=1}^{6} \hat{R}_z(\phi_j)\hat{H}\ket{0} \otimes \ket{b_7 \cdots b_{64}}\right),
\label{eq:toy_state}
\end{eqnarray}
where $\hat{V}_{64}$ is a fixed Clifford block acting on all $64$ qubits. Choose the measured Pauli observable $\hat{P}_{64}$ so that
\begin{eqnarray}
\hat{V}_{64}^\dagger \hat{P}_{64} \hat{V}_{64}
=
\hat{X}_1\hat{X}_2\hat{Y}_3\hat{X}_4\hat{X}_5\hat{Y}_6\hat{Z}_7\hat{Z}_8.
\label{eq:toy_heisenberg_image}
\end{eqnarray}
Such a choice is always possible because Clifford conjugation permutes the Pauli group.

The following proposition gives the exact projected feature and identifies its location relative to a low-order classical span.

\begin{widetext}
\begin{proposition}[Exact projected observable in the 64-qubit toy family]\label{prop:toy_exact}
For the state in Eq.~\eqref{eq:toy_state}, one has
\begin{eqnarray}
\bra{\Psi_x}\hat{P}_{64}\ket{\Psi_x} = (-1)^{b_7+b_8}\cos\phi_1\cos\phi_2\sin\phi_3\cos\phi_4\cos\phi_5\sin\phi_6.
\end{eqnarray}
Define
\begin{eqnarray}
q(x) := (-1)^{b_7+b_8}\cos\phi_1\cos\phi_2\sin\phi_3\cos\phi_4\cos\phi_5\sin\phi_6.
\label{eq:q_def}
\end{eqnarray}
Let $\mathcal{V}_{\le 7}$ be the span of all interaction-only trigonometric monomials of total degree at most $7$ in the raw variables
\begin{eqnarray}
z_7:=(-1)^{b_7},
\quad
z_8:=(-1)^{b_8},
\quad
\cos\phi_j,
\quad
\sin\phi_j
\quad (j=1,\ldots,6).
\label{eq:selected_raw_variables}
\end{eqnarray}
Under the product distribution in which $b_7,b_8$ are unbiased bits and $\phi_1,\ldots,\phi_6$ are independent and uniform on $[0,2\pi)$, the feature $q(x)$ is orthogonal in $L_2$ to $\mathcal{V}_{\le 7}$.
\end{proposition}
\end{widetext}

\begin{proof}---The expectation formula is the special case of {\bf Proposition~\ref{prop:ideal_score}} with
\begin{eqnarray}
A_X=\{1,2,4,5\}, ~~ A_Y=\{3,6\}, ~~ B_Z=\{7,8\},
\end{eqnarray}
applied to the Heisenberg image in Eq.~(\ref{eq:toy_heisenberg_image}). This proves the first claim.

For the orthogonality statement, write
\begin{eqnarray}
c_j:=\cos\phi_j, \quad s_j:=\sin\phi_j.
\end{eqnarray}
Then,
\begin{eqnarray}
q(x)=z_7 z_8 c_1 c_2 s_3 c_4 c_5 s_6,
\end{eqnarray}
which is the product of eight zero-mean generators. Let $m(x)$ be any interaction-only monomial of total degree at most $7$ in the variables listed in Eq.~(\ref{eq:selected_raw_variables}). Since $m$ has degree at most $7$, it omits at least one of the eight generators appearing in $q$.

If the omitted factor is $z_7$ or $z_8$, then $\E[qm]=0$ immediately, because the missing context bit has zero mean and is independent of all other factors. If the omitted factor is, say, $c_1$, then $m$ may contain neither factor from qubit $1$ or it may contain $s_1$, but in either case the $\phi_1$-integral in $\E[qm]$ contains an odd power of $\cos\phi_1$ and therefore vanishes over $[0,2\pi)$. The same argument applies when the omitted factor is one of $c_2,s_3,c_4,c_5,s_6$. Hence, $\E[qm]=0$ for every generator of $\mathcal V_{\le 7}$, so $q$ is orthogonal to the entire span.
\end{proof}

\begin{figure*}[t]
\centering
\includegraphics[width=0.80\textwidth]{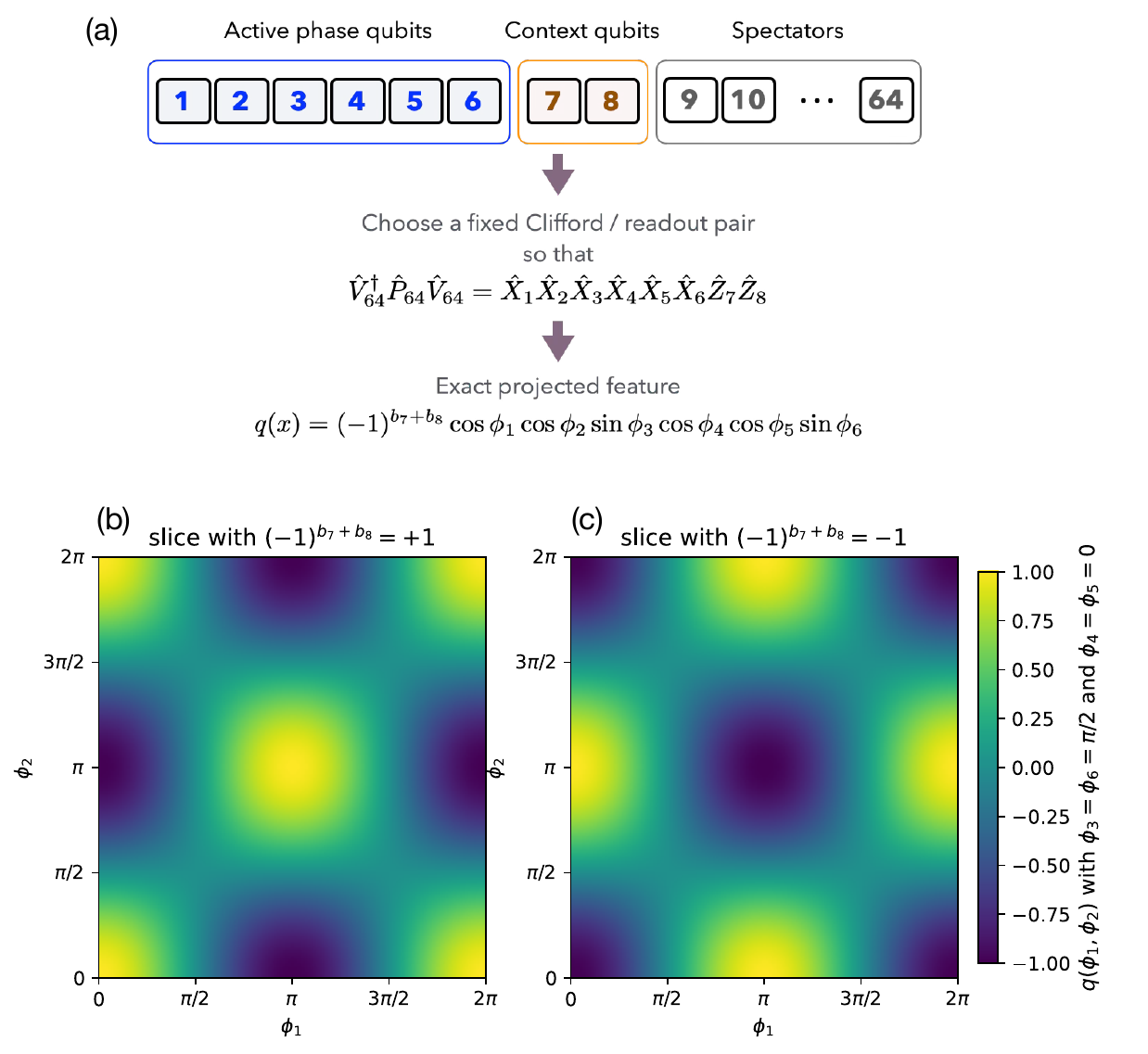}
\caption{Sixty-four-qubit toy family. Panel (a) shows the architecture of Eq.~\eqref{eq:toy_state}: qubits $1$--$6$ are the active phase qubits, qubits $7$ and $8$ are the informative context bits, and the remaining $56$ qubits stay on the classical path at data loading while still being available to the fixed Clifford block $\hat{V}_{64}$. Panels (b) and (c) show two-dimensional slices of the exact projected feature $q(x)$ for the two parities $(-1)^{b_7+b_8}=\pm 1$ after fixing $\phi_3=\phi_6=\pi/2$ and $\phi_4=\phi_5=0$. Thus the familiar $\cos\phi_1\cos\phi_2$ pattern survives as a visible slice of a much larger active-subspace model.}
\label{fig:toy_heatmap}
\end{figure*}

{\bf Proposition~\ref{prop:toy_exact}} is the direct large-register analogue of the residual-direction phenomenon from Sec.~\ref{sec:benefit}. A single projected observable compresses an eight-way interaction of six active phases and two context bits into one scalar feature, even though the full device contains $64$ qubits. Fig.~\ref{fig:toy_heatmap} visualizes this family. Here, Fig.~\ref{fig:toy_heatmap}(a) shows the active/context split; Fig.~\ref{fig:toy_heatmap}(b) and Fig.~\ref{fig:toy_heatmap}(c) show the two-dimensional slice obtained by fixing
\begin{eqnarray}
\phi_3=\phi_6=\frac{\pi}{2}, \quad \phi_4=\phi_5=0,
\end{eqnarray}
for which Eq.~(\ref{eq:q_def}) reduces to
\begin{eqnarray}
q(\phi_1,\phi_2)=\pm \cos\phi_1\cos\phi_2.
\end{eqnarray}

\subsection{A margin-controlled 64-qubit classification task and stronger classical baselines}

For the benchmark suite, we use the same $64$-qubit family, but we choose a data distribution with a controlled margin so that the projected feature does not become arbitrarily small. Specifically, let $\Delta:=\frac{\pi}{6}$. For the four $\hat{X}$-type active qubits $j \in \{1,2,4,5\}$, sample $\phi_j \in \{0,\pi\} + [-\Delta,\Delta]$ uniformly. For the two $\hat{Y}$-type active qubits $j \in \{3,6\}$, sample $\phi_j \in \left\{\frac{\pi}{2},\frac{3\pi}{2}\right\}+[-\Delta,\Delta]$ uniformly. Then, on the entire support one has
\begin{eqnarray}
\abs{\cos\phi_j} &\ge& \frac{\sqrt{3}}{2} \quad (j=1,2,4,5), \nonumber \\
\abs{\sin\phi_j} &\ge& \frac{\sqrt{3}}{2} \quad (j=3,6),
\end{eqnarray}
and therefore
\begin{eqnarray}
\abs{q(x)} \ge \left(\frac{\sqrt{3}}{2}\right)^6 \approx 0.422.
\label{eq:q_margin_bound}
\end{eqnarray}
This is the finite-size analogue of the margin-preserving regime discussed in Sec.~\ref{sec:scalable}: the projected score stays bounded away from zero even though the total register is large.

We define the clean label by
\begin{eqnarray}
y = \sgn\big(q(x)\big).
\label{eq:synthetic_label}
\end{eqnarray}
The task again is not meant to be a universal classical-versus-quantum separation. It is a controlled representation-efficiency experiment for a larger register. This distinction is important because recent discussions of practical quantum learning advantage stress comparison against strong task-specific classical models and explicit resource accounting~\cite{Schuld2022PRXQuantum,Liu2021NatPhys,Anschuetz2023,Agliardi2026}. We compare four explicit feature families in the main text.
\begin{itemize}
\item[(1)] A classical linear model on all raw variables $\Phi_{\mathrm{lin}} = \big( z_7, \ldots, z_{64}, \cos\phi_1, \sin\phi_1, \ldots, \cos\phi_6, \sin\phi_6\big)$, which has dimension $70$.
\item[(2)] A classical cubic interaction-only baseline on the selected $14$-variable set in Eq.~(\ref{eq:selected_raw_variables}). This already uses a nontrivial handcrafted expansion, but it still cannot contain the degree-eight target interaction.
\item[(3)] An exact degree-eight interaction-only baseline on the same selected $14$-variable set. This family contains the target monomial exactly, but only after a much larger engineered expansion.
\item[(4)] A hybrid projected model that augments $\Phi_{\mathrm{lin}}$ with the single quantum feature $q(x)$ from Eq.~(\ref{eq:q_def}).
\end{itemize}
Appendix~\ref{app:kernel_controls} adds degree-eight polynomial-kernel and RBF-kernel support-vector controls on the selected variables. Fig.~\ref{fig:learning_curves} collects the main explicit-feature comparison.

\begin{figure*}[t]
\centering
\includegraphics[width=0.98\textwidth]{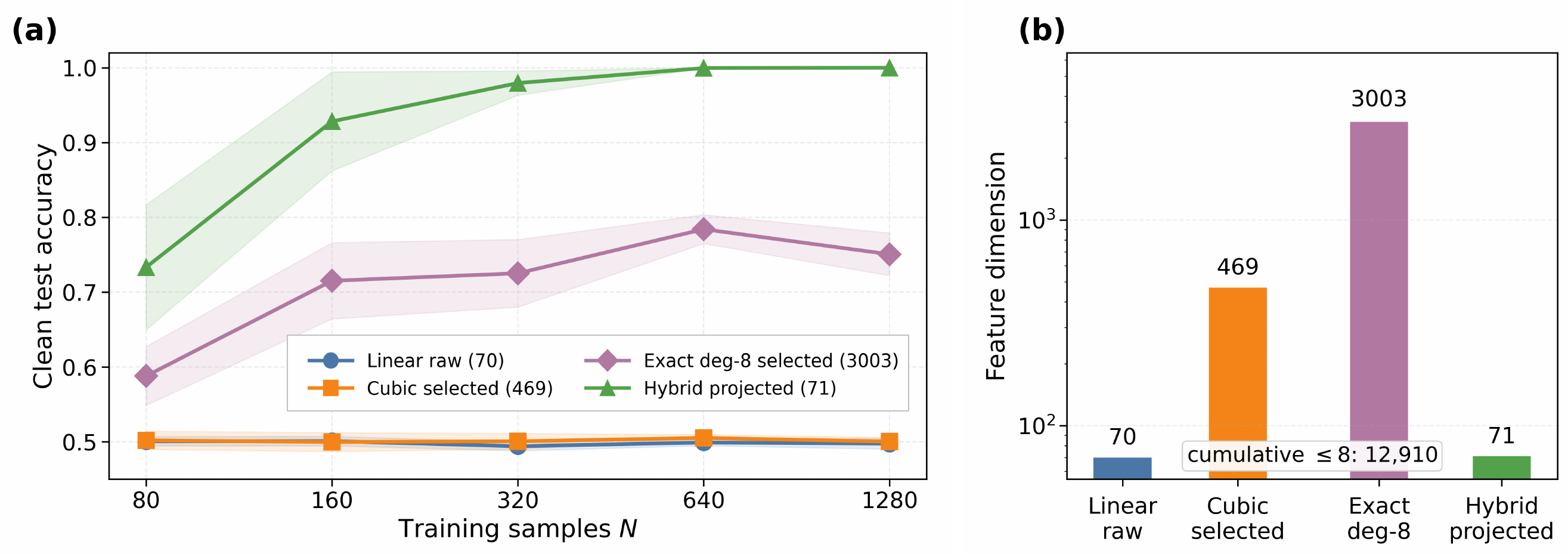}
\caption{Benchmark numerics for the $64$-qubit toy family. Panel (a) shows clean-test accuracy versus the number of training samples $N$ under $10\%$ symmetric label noise. The low-order classical baselines remain near chance. An exact degree-eight interaction-only baseline on the selected $14$ variables becomes nontrivial but still remains well below the hybrid projected model, which reaches essentially perfect accuracy by $N=640$. Panel (b) reports the explicit feature dimensions. The exact degree-eight selected family uses $3003$ features, while the cumulative interaction-only family up to degree $8$ would use $12{,}910$; Appendix~\ref{app:toy} records both counts.}
\label{fig:learning_curves}
\end{figure*}

For each training size $N$, we draw noisy training labels by flipping the clean label independently with probability $0.1$, fit ridge-type linear classifiers in the chosen explicit feature family, and evaluate on a clean test set. Fig.~\ref{fig:learning_curves}(a) shows the resulting clean-test accuracy curves. Across the tested range, the low-order linear and cubic baselines remain near chance. The exact degree-eight selected family becomes nontrivial but still lags the hybrid curve, rising from roughly $0.59$ mean clean-test accuracy at $N=80$ to about $0.78$ at $N=640$ in our runs, whereas the hybrid projected model rises from about $0.73$ at $N=80$ to essentially $1.00$ by $N=640$. Appendix~\ref{app:kernel_controls} shows that the tuned polynomial-kernel and RBF-kernel SVM controls also do not close this gap on the same task. The point is therefore not that exact or implicit classical recovery is impossible. It is that one projected quantum feature exposes the relevant high-order direction immediately, while exact classical recovery requires a much larger or less statistically efficient hypothesis family. Appendix~\ref{app:global_kernel} complements this large-register benchmark with a fully computable small-qubit comparison against a naive global fidelity kernel, where the projected AQSE kernel keeps the sample regularized dimension of order $M$ while the global kernel has sample regularized dimension of order $N$.

\subsection{Residual screening of candidate active subspaces}

To test the practical content of {\bf Proposition~\ref{prop:sample_gain}}, we generate a regression problem of the form
\begin{eqnarray}
y = f_C(x) + \alpha q_{S^\star}(x) + \xi,
\end{eqnarray}
where $f_C$ lies in the classical design span, $q_{S^\star}$ is a single projected quantum direction drawn from an eight-generator family, and $\xi$ is mean-zero Gaussian noise. For each candidate $4$-subset $S$ of the generator family, we build $q_S$, compute the finite-sample screening score
\begin{eqnarray}
\hat{s}(S) := \frac{1}{N}\left\|P_{Q_\perp(S)} r_C^{(N)}\right\|_2^2,
\end{eqnarray}
on the training split, and compare it with the independent test-risk gain obtained by adding $q_S$ to the classical regressor.

\begin{figure*}[t]
\centering
\includegraphics[width=0.80\textwidth]{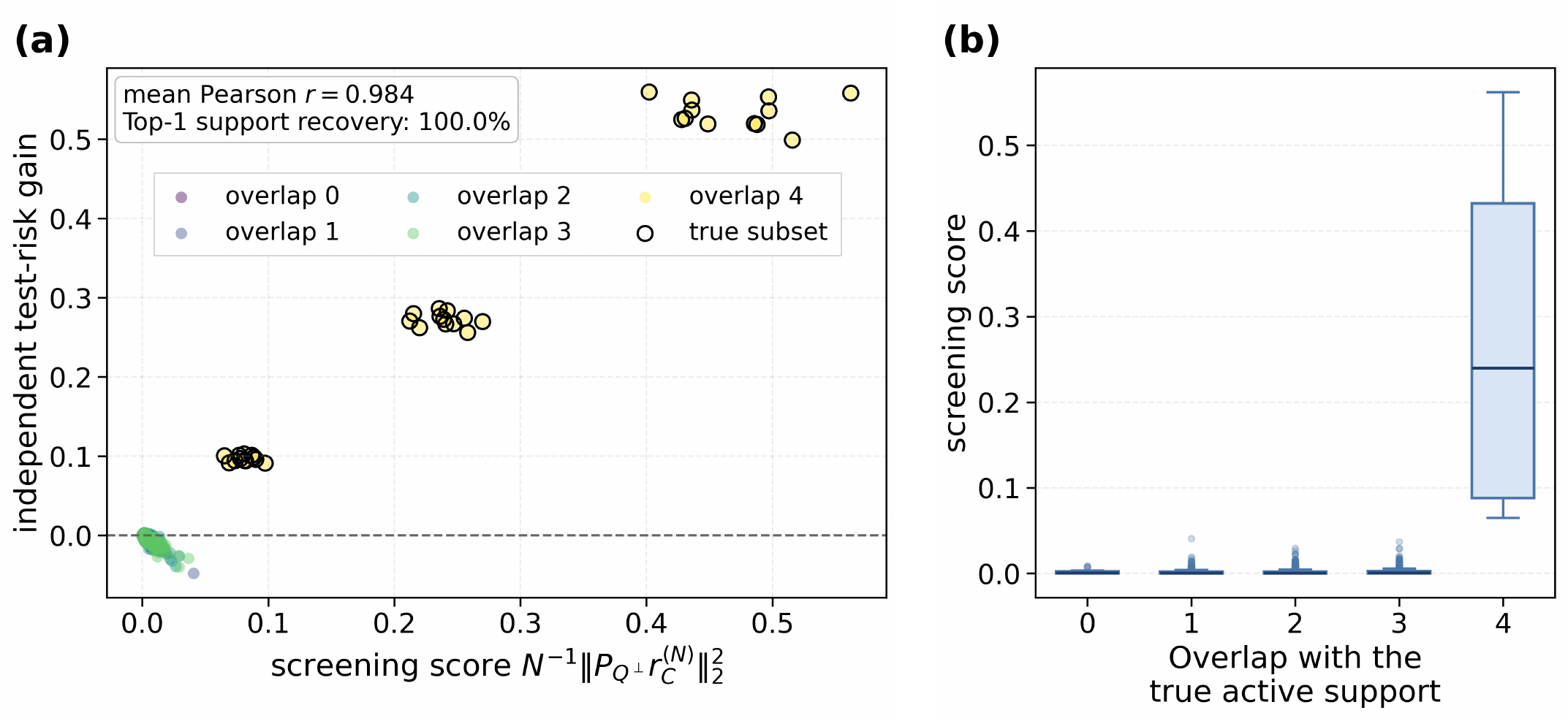}
\caption{Residual screening for candidate active subspaces. Panel (a) plots the finite-sample screening score from {\bf Proposition~\ref{prop:sample_gain}} against the independent test-risk gain obtained by adding the corresponding candidate feature to the classical regressor. Each point is a candidate $4$-subset from an eight-generator family; colors indicate its overlap with the true active support, and the exact support is marked explicitly. Panel (b) shows the distribution of the screening scores versus support overlap. Across the runs used here, the mean Pearson correlation between the screening score and the observed test-risk gain is $0.984$, and the exact support is selected as the top-scoring candidate in every run.}
\label{fig:screening_residual}
\end{figure*}

Fig.~\ref{fig:screening_residual} shows that the empirical screening score and the observed test-risk gain are tightly aligned: the mean Pearson correlation across runs is $0.984$, and the exact active subset is recovered as the top-scoring candidate in all runs of this synthetic family. The dependence on support overlap is equally sharp. In this sense, Appendix~\ref{app:residual_protocol} is not merely a structural identity; it also suggests a usable residual-based rule for choosing which candidate active subspace deserves quantum resources.

\subsection{Noise scaling and \texorpdfstring{$\beta^{-2}$}{beta^-2} sample demand}

We next test the sample-complexity mechanism of Secs.~\ref{sec:noisy} and \ref{sec:scalable} using the explicit logarithmic-support family of Appendix~\ref{app:explicit_family}. To isolate the dependence on the oracle reliability floor, we fix the projected-score magnitude at its designed floor and vary the active support $a$ and the light-cone attenuation $\Lambda$ so that
\begin{eqnarray}
\beta_0 = \Lambda c_0^a.
\end{eqnarray}
For each scenario, we query noisy labels with correctness probability $(1+\beta_0)/2$ and train a one-feature linear separator on the projected score.

\begin{figure*}[t]
\centering
\includegraphics[width=0.95\textwidth]{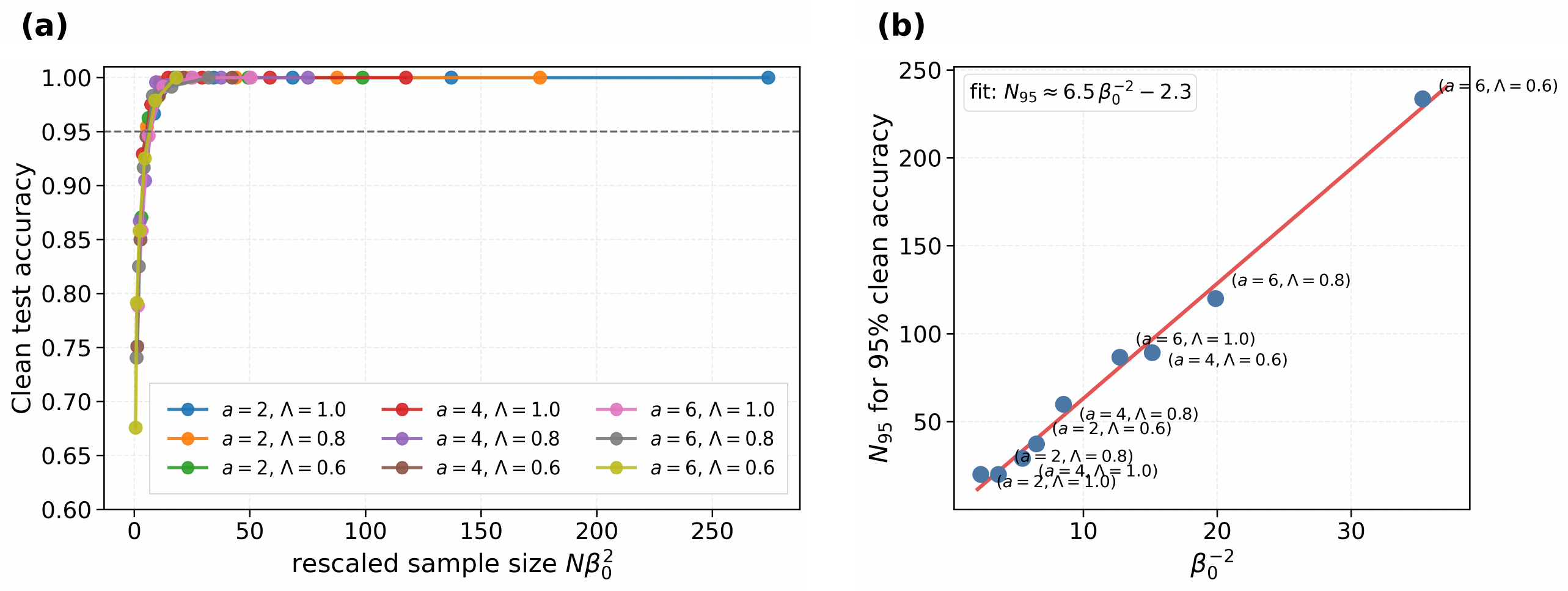}
\caption{Noise scaling and the $\beta_0^{-2}$ sample demand. Panel (a) shows clean-test accuracy for several $(a,\Lambda)$ pairs after rescaling the horizontal axis by $N\beta_0^2$, where $\beta_0=\Lambda c_0^a$ is the oracle reliability floor. The learning curves nearly collapse after this rescaling. Panel (b) reports the sample size $N_{95}$ needed to reach $95\%$ clean accuracy versus $\beta_0^{-2}$; the fitted line is approximately linear in $\beta_0^{-2}$, in agreement with {\bf Theorem~\ref{thm:noisy_pac}}.}
\label{fig:beta_scaling}
\end{figure*}

The collapse in Fig.~\ref{fig:beta_scaling} is the numerical counterpart of {\bf Theorem~\ref{thm:noisy_pac}}: once the horizontal axis is rescaled by $N\beta_0^2$, the learning curves nearly coincide, and the sample size needed to reach $95\%$ clean accuracy is approximately linear in $\beta_0^{-2}$. In our calibration, scenarios with $\beta_0\approx 0.65$ reach $95\%$ accuracy with about $20$ samples, whereas a smaller reliability $\beta_0\approx 0.17$ requires about $234$ samples. Appendix~\ref{app:finite_shot} then adds a finite-shot projected-readout study for the $64$-qubit toy family under combined dephasing and readout noise.

The large-register benchmark suite therefore makes the same conceptual point as the original three-qubit example, but in a form that is much closer to the full narrative of the paper. The active coherent sector remains tiny, the projected readout remains low-dimensional, the residual-screening diagnostic is operational, and the noise-limited sample demand follows the predicted $\beta^{-2}$ law. Appendix~\ref{app:toy}, Appendix~\ref{app:global_kernel}, and Appendix~\ref{app:finite_shot} record the corresponding protocol details, additional kernel controls, the global-kernel comparison, and the finite-shot simulations.

\section{Discussion and outlook}\label{sec:discussion}

The main scientific claim of our work can be summarized in one sentence: full quantum data-encoding is not necessary for scalable hybrid quantum learning advantage. What is necessary is much more specific. One needs a projected quantum sector that contributes a direction outside the classical span, a controlled sample regularized dimension, and a noisy oracle reliability that does not collapse too quickly.

The framework developed here organizes these requirements into a coherent hierarchy. {\bf Theorem~\ref{thm:kernel_rank}} controls the size of the projected representation. {\bf Theorem~\ref{thm:residual}} identifies the precise condition under which the quantum sector adds real predictive value. {\bf Theorem~\ref{thm:noisy_pac}} turns the oracle reliability into a sample-complexity statement. {\bf Theorem~\ref{thm:dephasing}} and {\bf Corollary~\ref{cor:poly_persistence}} then show that, at least for a canonical family, this reliability can remain inverse-polynomial even when the number of qubits and the encoding gate count both grow polynomially.


The present theory is intentionally focused rather than exhaustive. First, our strongest scalability result is proved for a canonical Clifford family with projected Pauli readout under local dephasing; this already gives an explicit scalable regime, while extending the exact analysis to broader circuit classes, more general Pauli-diagonal noise, or non-Clifford but approximately stable observable sectors is a natural next step. Second, the benefit theorem in Sec.~\ref{sec:benefit} is a learning-theoretic characterization rather than a complexity-theoretic separation. It tells us when the hybrid representation is better than the chosen classical span, though it does not rule out every efficient classical imitation of the projected quantum feature. This caveat is important in light of recent work emphasizing that quantum learning advantage claims should be tied to concrete resources, data geometry, and classical baselines~\cite{Schuld2022PRXQuantum,Liu2021NatPhys,Huang2022Science,Anschuetz2023,King2024,Liu2025Science}. Third, our numerics are still mechanism-driven rather than a claim of universal classical separation. They now include stronger classical baselines, a projected-versus-global-kernel comparison, residual-screening tests, $\beta^{-2}$ sample-scaling curves, and finite-shot noisy-readout simulations, but they remain synthetic and designed to isolate the structural statements of the paper rather than to exhaust every possible classical heuristic.

These extensions do not alter the main conclusion of the present paper. The purpose of this work is not to settle every aspect of quantum learning advantage. It is to define a scalable and analytically tractable route by which such an advantage can arise without full quantum data-encoding. The framework suggests several concrete follow-up problems. One direction is to search for active-subspace encoders in which the projected observables remain noise-stable while the induced functions become harder to approximate classically. Another direction is to study data-dependent methods for choosing the active subset itself. {\bf Theorem~\ref{thm:residual}} suggests that the right subset is the one whose projected quantum features correlate with the residual left by the classical path. Turning this observation into an adaptive architecture-selection principle would be practically important.


In conclusion, one should apply quantum data-encoding only where it is actually needed. If the selected active subspace carries information outside the classical span and the projected readout remains statistically and noise-wise stable, then quantum learning advantage can persist without full quantum data-encoding and without super-polynomial encoding overhead. This provides a principled foundation for a scalable, QRAM-free, and NISQ-compatible program of hybrid quantum learning.

\section*{Acknowledgements}

This work was supported by the Ministry of Science, ICT and Future Planning (MSIP) through the National Research Foundation of Korea (RS-2024-00432214, RS-2025-03532992, and RS-2025-18362970) and the Institute of Information and Communications Technology Planning and Evaluation grant funded by the Korean government (RS-2019-II190003, ``Research and Development of Core Technologies for Programming, Running, Implementing and Validating of Fault-Tolerant Quantum Computing System''), the Korean ARPA-H Project through the Korea Health Industry Development Institute (KHIDI), funded by the Ministry of Health \& Welfare, Republic of Korea (RS-2025-25456722), and the Ministry of Trade, Industry and Resources (MOTIR), Korea, under the project ``Industrial Technology Infrastructure Program'' (RS-2024-00466693). This work was also supported by Korea Institute of Science and Technology Information (KISTI) ((KISTI)K26L1M3C5-01). We acknowledge the Yonsei University Quantum Computing Project Group for providing support and access to the Quantum System One (Eagle Processor), which is operated at Yonsei University.

\onecolumngrid

\appendix
\setcounter{theorem}{0}
\setcounter{proposition}{0}
\setcounter{lemma}{0}
\setcounter{corollary}{0}
\setcounter{definition}{0}
\setcounter{remark}{0}
\renewcommand{\thetheorem}{\Alph{section}.\arabic{theorem}}
\renewcommand{\theproposition}{\Alph{section}.\arabic{proposition}}
\renewcommand{\thelemma}{\Alph{section}.\arabic{lemma}}
\renewcommand{\thecorollary}{\Alph{section}.\arabic{corollary}}
\renewcommand{\thedefinition}{\Alph{section}.\arabic{definition}}
\renewcommand{\theremark}{\Alph{section}.\arabic{remark}}
\renewcommand{\theHtheorem}{appendix.\Alph{section}.\arabic{theorem}}
\renewcommand{\theHproposition}{appendix.\Alph{section}.\arabic{proposition}}
\renewcommand{\theHlemma}{appendix.\Alph{section}.\arabic{lemma}}
\renewcommand{\theHcorollary}{appendix.\Alph{section}.\arabic{corollary}}
\renewcommand{\theHdefinition}{appendix.\Alph{section}.\arabic{definition}}
\renewcommand{\theHremark}{appendix.\Alph{section}.\arabic{remark}}

\section{Exact residual-gain formulas and a finite-sample criterion}\label{app:residual_exact}

This appendix records two useful refinements of {\bf Theorem~\ref{thm:residual}}. The first gives the exact gain from adding the projected quantum sector, rather than only a one-direction lower bound. The second gives the corresponding finite-sample matrix formula, which can be used directly when one screens candidate active quantum subspaces on a data set.

\begin{proposition}[Exact gain from the orthogonalized quantum sector]\label{prop:exact_gain}
Let $\mathcal{V}_C \subset L_2(P_X)$ and $\mathcal{V}_Q \subset L_2(P_X)$ be as in Sec.~\ref{sec:benefit}, and let
\begin{eqnarray}
\mathcal{W} := (\id - P_C)\mathcal{V}_Q \subset \mathcal{V}_C^{\perp}.
\end{eqnarray}
Then,
\begin{eqnarray}
\mathcal{V}_H = \mathcal{V}_C \oplus \mathcal{W}
\end{eqnarray}
as an orthogonal direct sum, and the optimal risk reduction satisfies
\begin{eqnarray}
R_C^* - R_H^* = \norm{P_{\mathcal W} r_C}_2^2,
\label{eq:exact_gain_formula_appendix}
\end{eqnarray}
where $r_C=Y-P_CY$ is the classical residual. In particular, $R_H^*<R_C^*$ if and only if $P_{\mathcal{W}} r_C \neq 0$.
\end{proposition}

\begin{proof}---By definition, $(\id - P_C)\mathcal{V}_Q \subset \mathcal{V}_C^{\perp}$, so every element of $\mathcal{W}$ is orthogonal to $\mathcal{V}_C$. We now show that every element of $\mathcal{V}_H = \mathcal{V}_C + \mathcal{V}_Q$ can be written uniquely as a sum of one vector in $\mathcal{V}_C$ and one vector in $\mathcal{W}$.

Let $v \in \mathcal{V}_H$. Then, $v=c+q$ for some $c \in \mathcal{V}_C$ and $q \in \mathcal{V}_Q$. Decompose
\begin{eqnarray}
q = P_C q +(\id - P_C)q.
\end{eqnarray}
Hence,
\begin{eqnarray}
v = \big(c + P_C q\big) + (\id - P_C)q,
\end{eqnarray}
where $c + P_C q \in \mathcal{V}_C$ and $(\id - P_C)q \in \mathcal{W}$. Thus, $\mathcal{V}_H \subseteq \mathcal{V}_C+\mathcal{W}$, while the reverse inclusion is immediate from the definition of $\mathcal{W}$. Therefore, $\mathcal{V}_H = \mathcal{V}_C + \mathcal{W}$. Since $\mathcal{W} \subset \mathcal{V}_C^{\perp}$, the sum is orthogonal, so
\begin{eqnarray}
\mathcal{V}_H = \mathcal{V}_C \oplus \mathcal{W}.
\end{eqnarray}

The orthogonal projector onto $\mathcal{V}_H$ is therefore
\begin{eqnarray}
P_H = P_C + P_{\mathcal{W}}.
\end{eqnarray}
Applying this to $Y = P_C Y + r_C$ and using $r_C \perp \mathcal{V}_C$ gives
\begin{eqnarray}
P_H Y = P_C Y + P_{\mathcal{W}} r_C.
\end{eqnarray}
Consequently,
\begin{eqnarray}
Y - P_H Y = r_C - P_{\mathcal{W}} r_C.
\end{eqnarray}
Because $P_{\mathcal{W}} r_C$ is the orthogonal projection of $r_C$ onto $\mathcal{W}$, the Pythagorean theorem yields
\begin{eqnarray}
\norm{r_C}_2^2 = \norm{r_C - P_{\mathcal{W}} r_C}_2^2 + \norm{P_{\mathcal{W}} r_C}_2^2.
\end{eqnarray}
Recalling that $R_C^*=\norm{r_C}_2^2$ and $R_H^*=\norm{Y - P_H Y}_2^2=\norm{r_C - P_{\mathcal{W}} r_C}_2^2$, we obtain Eq.~(\ref{eq:exact_gain_formula_appendix}). Strict improvement holds exactly when the last term is nonzero.
\end{proof}

{\bf Proposition~\ref{prop:exact_gain}} shows that the projected quantum sector acts only through its orthogonalized component $\mathcal{W} = (\id - P_C)\mathcal{V}_Q$. In other words, once the classical span is fixed, the entire hybrid gain is concentrated in the projection of the classical residual onto this orthogonalized quantum subspace.

\begin{proposition}[Finite-sample matrix criterion]\label{prop:sample_gain}
Let $y \in \mathbb R^N$ be a response vector, let $F_C \in \mathbb R^{N \times d_C}$ be a classical design matrix, and let $F_Q \in \mathbb R^{N \times M}$ be a projected quantum feature matrix. Let $P_C$ be the orthogonal projector onto $\mathrm{col}(F_C)$ and let
\begin{eqnarray}
Q_{\perp}:=(\id_N - P_C)F_Q.
\end{eqnarray}
Denote by $P_{Q_{\perp}}$ the orthogonal projector onto $\mathrm{col}(Q_{\perp})$, and by $P_H$ the orthogonal projector onto $\mathrm{col}([F_C \; F_Q])$. Then, the empirical least-squares risks
\begin{eqnarray}
\widehat R_C &:=& \frac{1}{N}\norm{(I_N-P_C)y}_2^2, \nonumber \\
\widehat R_H &:=& \frac{1}{N}\norm{(I_N-P_H)y}_2^2
\end{eqnarray}
satisfy
\begin{eqnarray}
\widehat{R}_C - \widehat{R}_H = \frac{1}{N}\norm{P_{Q_{\perp}}r_C^{(N)}}_2^2,
\label{eq:sample_gain_formula}
\end{eqnarray}
where $r_C^{(N)}:=(I_N-P_C)y$. In particular, $\widehat{R}_H < \widehat{R}_C$ if and only if $P_{Q_{\perp}} r_C^{(N)} \neq 0$.
\end{proposition}

\begin{proof}---The argument is the finite-dimensional version of {\bf Proposition~\ref{prop:exact_gain}}. Every column of $Q_{\perp}$ lies in $\mathrm{col}(F_C)^{\perp}$ by construction. Moreover,
\begin{eqnarray}
\mathrm{col}([F_C \; F_Q]) = \mathrm{col}(F_C) \oplus \mathrm{col}(Q_{\perp}),
\end{eqnarray}
because each column of $F_Q$ decomposes as $P_C f + (\id_N - P_C)f$, with the first term in $\mathrm{col}(F_C)$ and the second term in $\mathrm{col}(Q_{\perp})$. Therefore,
\begin{eqnarray}
P_H = P_C + P_{Q_{\perp}}.
\end{eqnarray}
Applying this to $y=P_C y + r_C^{(N)}$ gives
\begin{eqnarray}
P_H y = P_C y + P_{Q_{\perp}} r_C^{(N)}.
\end{eqnarray}
Hence,
\begin{eqnarray}
(\id_N - P_H)y = r_C^{(N)} - P_{Q_{\perp}} r_C^{(N)}.
\end{eqnarray}
Because $P_{Q_{\perp}} r_C^{(N)}$ is the orthogonal projection of $r_C^{(N)}$ onto $\mathrm{col}(Q_{\perp})$, the Pythagorean theorem yields
\begin{eqnarray}
\norm{r_C^{(N)}}_2^2 = \norm{r_C^{(N)} - P_{Q_{\perp}}r_C^{(N)}}_2^2 + \norm{P_{Q_{\perp}}r_C^{(N)}}_2^2.
\end{eqnarray}
Dividing by $N$ proves Eq.~(\ref{eq:sample_gain_formula}).
\end{proof}

{\bf Proposition~\ref{prop:sample_gain}} suggests a direct empirical screening rule for active quantum subspaces: one should prefer the projected quantum feature family whose orthogonalized sample matrix $Q_{\perp}$ captures the largest component of the classical residual.

\section{Uniform convergence bound used in Theorem~\ref{thm:noisy_pac}}\label{app:VC}

For completeness, we prove the VC-type inequality used in Eq.~(\ref{eq:VC_bound}). The proof is standard, but we include it here so that the PAC theorem of Sec.~\ref{sec:noisy} is logically self-contained.

Let
\begin{eqnarray}
\mathcal{L}_{\mathcal{H}} := \big\{\ell_h : (x,\widetilde y)\mapsto \mathds{I}\{h(x)\neq \widetilde y\} : h \in \mathcal{H}\big\}
\end{eqnarray}
be the binary loss class induced by $\mathcal{H}$.

\begin{lemma}[Growth function of the induced loss class]\label{lem:growth_loss}
For every integer $m \ge 1$,
\begin{eqnarray}
\Pi_{\mathcal{L}_{\mathcal{H}}}(m) \le \Pi_{\mathcal{H}}(m),
\end{eqnarray}
where $\Pi_{\mathcal{H}}(m)$ is the growth function of the hypothesis class $\mathcal{H}$ on unlabeled inputs.
\end{lemma}

\begin{proof}---Fix labeled points $z_i=(x_i,\widetilde{y}_i)$, $i=1,\dots,m$. For any $h \in \mathcal{H}$, the loss vector on these points is
\begin{eqnarray}
\big(\ell_h(z_1),\dots,\ell_h(z_m)\big) = \big(\mathds{I}\{h(x_1) \neq \widetilde{y}_1\},\dots,\mathds{I}\{h(x_m)\neq \widetilde{y}_m\}\big).
\end{eqnarray}
Since the labels $\widetilde{y}_i$ are fixed, the map from the prediction vector $(h(x_1),\dots,h(x_m))$ to the loss vector is one-to-one. Therefore, the number of distinct loss vectors realized on the labeled sample cannot exceed the number of distinct prediction vectors realized by $\mathcal{H}$ on the input sample $(x_1,\dots,x_m)$. Taking the supremum over all labeled samples, we can complete the proof.
\end{proof}

\begin{lemma}[Uniform convergence for noisy empirical risk]\label{lem:vc_uniform}
Let $\mathcal{H}$ be a binary hypothesis class and let
\begin{eqnarray}
\widehat{R}_\eta(h) = \frac{1}{N}\sum_{i=1}^{N}\ell_h(Z_i),
\quad
R_\eta(h)=\E[\ell_h(Z)],
\end{eqnarray}
where $Z_i=(X_i,\widetilde{Y}_i)$ are i.i.d. Then, for every $\alpha>0$,
\begin{eqnarray}
\Prob\left( \sup_{h \in \mathcal{H}}\abs{\widehat{R}_\eta(h)-R_\eta(h)} > \alpha \right) \le 4\Pi_{\mathcal H}(2N) e^{-\frac{N\alpha^2}{8}}.
\label{eq:appendix_vc_general}
\end{eqnarray}
If $\mathcal{H}$ has VC dimension $d \ge 1$, then by Sauer--Shelah,
\begin{eqnarray}
\Prob\left( \sup_{h \in \mathcal{H}}\abs{\widehat{R}_\eta(h)-R_\eta(h)} > \alpha \right) \le 4\left(\frac{2eN}{d}\right)^d e^{-\frac{N\alpha^2}{8}}.
\label{eq:appendix_vc_final}
\end{eqnarray}
\end{lemma}

\begin{proof}---Let $Z_1',\dots,Z_N'$ be an independent ghost sample with the same distribution as $Z_1,\dots,Z_N$, and write
\begin{eqnarray}
\widehat{R}_\eta'(h) := \frac{1}{N}\sum_{i=1}^{N}\ell_h(Z_i').
\end{eqnarray}
We begin with symmetrization. If $\sup_h\abs{\widehat{R}_\eta(h) - R_\eta(h)} > \alpha$ and simultaneously $\sup_h\abs{\widehat{R}_\eta'(h) - R_\eta(h)} \le \frac{\alpha}{2}$, then necessarily $\sup_h\abs{\widehat{R}_\eta(h) - \widehat{R}_\eta'(h)} > \frac{\alpha}{2}$. Therefore, by a union bound and because the primed and unprimed samples have the same law,
\begin{eqnarray}
\Prob\left(\sup_h\abs{\widehat{R}_\eta(h) - R_\eta(h)} > \alpha\right) \le 2\Prob\left(\sup_h\abs{\widehat{R}_\eta(h) - \widehat{R}_\eta'(h)} > \frac{\alpha}{2}\right).
\label{eq:symmetrization_appendix}
\end{eqnarray}

Condition now on the combined sample
\begin{eqnarray}
(Z_1,Z_1',\dots,Z_N,Z_N').
\end{eqnarray}
For each $i$, introduce an independent Rademacher variable $\sigma_i \in \{-1,+1\}$. By exchangeability of $(Z_i, Z_i')$, the conditional distribution of
\begin{eqnarray}
\widehat{R}_\eta(h) - \widehat{R}_\eta'(h) = \frac{1}{N}\sum_{i=1}^{N}\big(\ell_h(Z_i)-\ell_h(Z_i')\big)
\end{eqnarray}
is the same as that of
\begin{eqnarray}
\frac{1}{N}\sum_{i=1}^{N}\sigma_i\big(\ell_h(Z_i)-\ell_h(Z_i')\big).
\end{eqnarray}
For a fixed $h$, the summands are independent, mean-zero, and lie in $[-1,1]$. Hence, Hoeffding's inequality gives
\begin{eqnarray}
\Prob_\sigma\left( \abs{\frac{1}{N}\sum_{i=1}^{N}\sigma_i\big(\ell_h(Z_i) - \ell_h(Z_i')\big)} > \frac{\alpha}{2} \right) \le 2e^{-\frac{N\alpha^2}{8}}.
\label{eq:hoeffding_appendix}
\end{eqnarray}

On the fixed combined sample, the number of distinct loss vectors of the class $\mathcal L_{\mathcal H}$ on the $2N$ labeled points is at most $\Pi_{\mathcal L_{\mathcal H}}(2N)$. Therefore, applying a union bound over all distinct loss vectors and then {\bf Lemma~\ref{lem:growth_loss}}, we obtain the conditional estimate
\begin{eqnarray}
\Prob_\sigma\left( \sup_{h\in\mathcal H} \abs{\frac{1}{N}\sum_{i=1}^{N}\sigma_i\big(\ell_h(Z_i)-\ell_h(Z_i')\big)} > \frac{\alpha}{2} \right) \le 2\Pi_{\mathcal H}(2N)e^{-\frac{N\alpha^2}{8}}.
\end{eqnarray}
Averaging over the combined sample and combining with Eq.~(\ref{eq:symmetrization_appendix}) gives Eq.~(\ref{eq:appendix_vc_general}). If the VC dimension of $\mathcal{H}$ is $d \ge 1$, Sauer--Shelah yields
\begin{eqnarray}
\Pi_{\mathcal H}(2N)\le \left(\frac{2eN}{d}\right)^d,
\end{eqnarray}
which proves Eq.~(\ref{eq:appendix_vc_final}).
\end{proof}

{\bf Lemma~\ref{lem:vc_uniform}} is exactly the estimate used in the proof of {\bf Theorem~\ref{thm:noisy_pac}}.

\section{Extension to general local Pauli-diagonal noise}\label{app:noise_general}

The dephasing analysis of Sec.~\ref{sec:scalable} becomes more transparent when written in the language of Pauli-diagonal channels. This appendix shows that the same Clifford-Heisenberg argument works verbatim for a much larger family of noise models.

Consider a Clifford circuit $\hat{V}_n = \hat{U}_G \cdots \hat{U}_1$ acting on $\kappa(n)$ qubits. After each ideal gate layer $\hat{U}_g$, let the device undergo a product noise channel
\begin{eqnarray}
\mathcal{N}_g := \bigotimes_{j=1}^{\kappa(n)} \mathcal N_{g,j},
\end{eqnarray}
where each one-qubit channel is \emph{Pauli diagonal} in the Heisenberg picture:
\begin{eqnarray}
\mathcal{N}_{g,j}^\ast(\hat{\id})=\hat{\id},
\quad
\mathcal{N}_{g,j}^\ast(\hat{\sigma})=\lambda_{g,j}(\hat{\sigma})\hat{\sigma},
\label{eq:pauli_diagonal_local}
\end{eqnarray}
where $\hat{\sigma} \in \{\hat{X},\hat{Y},\hat{Z}\}$ with $\abs{\lambda_{g,j}(\hat{\sigma})} \le 1$.

For a measured Pauli observable $\hat{P}_n$, define $\hat{\sigma}_{g,j}(\hat{P}_n) \in \{\hat{\id},\hat{X},\hat{Y},\hat{Z}\}$ to be the one-qubit Pauli factor on qubit $j$ obtained after propagating $\hat{P}_n$ backward through the ideal suffix $\hat{U}_G\cdots \hat{U}_{g+1}$. Because the circuit is Clifford, this factor is well defined.

\begin{theorem}[Exact score attenuation under local Pauli-diagonal noise]\label{thm:pauli_general}
Let $\hat{\rho}_x^{\mathrm{noisy}}$ be the output of the canonical Clifford AQSE family when the ideal circuit is interleaved with the local Pauli-diagonal channels above. Then, for every input $x$,
\begin{eqnarray}
\Tr[\hat{P}_n\hat{\rho}_x^{\mathrm{noisy}}] = \Lambda_n^{\mathrm{PD}}(\hat{P}_n)\,s_n(x),
\label{eq:pauli_general_formula}
\end{eqnarray}
where $s_n(x)$ is the ideal score and
\begin{eqnarray}
\Lambda_n^{\mathrm{PD}}(\hat{P}_n) := \prod_{g=1}^{G}\prod_{j=1}^{\kappa(n)} \lambda_{g,j}\big(\hat{\sigma}_{g,j}(\hat{P}_n)\big),
\label{eq:pauli_general_lambda}
\end{eqnarray}
with the convention $\lambda_{g,j}(\hat{I})=1$.
\end{theorem}

\begin{proof}---Work in the Heisenberg picture. Start from the measured Pauli string $\hat{P}_n$ at the output. Since each $\hat{U}_g$ is Clifford, conjugation by $\hat{U}_g^\dagger$ maps a Pauli string to another Pauli string. Therefore, just before the action of the noise channel $\mathcal{N}_g^\ast$, the observable is a Pauli string whose one-qubit factor on qubit $j$ equals $\hat{\sigma}_{g,j}(\hat{P}_n)$. Applying Eq.~(\ref{eq:pauli_diagonal_local}) multiplies this factor by $\lambda_{g,j}(\hat{\sigma}_{g,j}(\hat{P}_n))$. Repeating this for all layers and qubits yields
\begin{eqnarray}
\mathcal C_n^\ast(\hat{P}_n) = \left[ \prod_{g=1}^{G}\prod_{j=1}^{\kappa(n)}\lambda_{g,j}\big(\hat{\sigma}_{g,j}(\hat{P}_n)\big) \right] \hat{Q}_n,
\end{eqnarray}
where $\hat{Q}_n=\hat{V}_n^\dagger\hat{P}_n\hat{V}_n$ is the ideal Heisenberg image. Taking expectation in the encoded state gives Eq.~(\ref{eq:pauli_general_formula}).
\end{proof}

\begin{corollary}[Inverse-polynomial persistence from logarithmic total attenuation]\label{cor:pauli_general_poly}
Suppose that for every relevant factor one has
\begin{eqnarray}
\lambda_{g,j}\big(\hat{\sigma}_{g,j}(\hat{P}_n)\big) \ge 1-\epsilon_{g,j}, \quad \left( 0 \le \epsilon_{g,j}\le \frac{1}{2} \right),
\end{eqnarray}
and let
\begin{eqnarray}
\Gamma_n(\hat{P}_n) := \sum_{g=1}^{G}\sum_{j=1}^{\kappa(n)} \epsilon_{g,j} \mathds{I}\{\hat{\sigma}_{g,j}(\hat{P}_n)\neq \hat{\id}\}.
\end{eqnarray}
Then,
\begin{eqnarray}
\Lambda_n^{\mathrm{PD}}(\hat{P}_n)\ge e^{-2\Gamma_n(\hat{P}_n)}.
\label{eq:gamma_bound}
\end{eqnarray}
In particular, if $\Gamma_n(\hat{P}_n)=O(\log n)$, then $\Lambda_n^{\mathrm{PD}}(\hat{P}_n)\ge n^{-O(1)}$.
\end{corollary}

\begin{proof}---Using Eq.~(\ref{eq:pauli_general_lambda}) and the assumption $\epsilon_{g,j} \le 1/2$, we obtain $\Lambda_n^{\mathrm{PD}}(\hat{P}_n) \ge \prod_{(g,j) : \hat{\sigma}_{g,j}(\hat{P}_n) \neq \hat{\id}}(1-\epsilon_{g,j})$. For $0 \le t \le 1/2$, the bound $\log(1-t) \ge -2t$ holds. Therefore, $\log\Lambda_n^{\mathrm{PD}}(\hat{P}_n) \ge -2\sum_{(g,j):\hat{\sigma}_{g,j}(\hat{P}_n)\neq \hat{I}}\epsilon_{g,j} = -2\Gamma_n(\hat{P}_n)$, which proves Eq.~(\ref{eq:gamma_bound}). The inverse-polynomial statement follows immediately when $\Gamma_n(\hat{P}_n)=O(\log n)$.
\end{proof}

The dephasing model of Sec.~\ref{sec:scalable} is the special case in which $\lambda_{g,j}(\hat{X})=\lambda_{g,j}(\hat{Y})=1-2p_g$ and $\lambda_{g,j}(\hat{Z})=1$.

\section{An explicit scalable family with logarithmic active support}\label{app:explicit_family}

We now package the previous results into one explicit family that realizes the main asymptotic message of the paper. The family is intentionally simple. Its purpose is not to prove a classical hardness separation, but to show concretely that polynomial encoding cost and growing qubit number are compatible with inverse-polynomial hybrid reliability when only a logarithmic active subset is quantum-encoded.

\begin{theorem}[Explicit logarithmic-support AQSE family]\label{thm:explicit_family}
Fix constants $0 < \phi_0 < \pi/2$ and $\gamma > 0$. For each $n \ge 2$, define
\begin{eqnarray}
a(n) &:=& \lceil \gamma \log n \rceil, \nonumber \\
\kappa(n) &:=& n, \nonumber \\
S_n &:=& \{1,2,\dots,a(n)\}.
\end{eqnarray}
Consider the encoder family on $\kappa(n)$ qubits given by
\begin{eqnarray}
\hat{E}_n(x) := \left[\bigotimes_{j=1}^{a(n)}\hat{R}_z(\phi_j(x))\hat{H}_j\right] \otimes \left[\bigotimes_{j=a(n)+1}^{n} \hat{\id}_j\right],
\label{eq:explicit_encoder}
\end{eqnarray}
where the data support satisfies $\abs{\phi_j(x)} \le \phi_0$ for all active qubits. Let the variational block be any Clifford circuit $\hat{V}_n$ for which there exists a measured Pauli observable $\hat{P}_n$ whose ideal Heisenberg image has the form
\begin{eqnarray}
\hat{Q}_n = \hat{V}_n^\dagger \hat{P}_n \hat{V}_n = \pm \left[\bigotimes_{j=1}^{a(n)}\hat{X}_j\right] \otimes \hat{R}_n,
\label{eq:explicit_Qn}
\end{eqnarray}
where $\hat{R}_n$ is a Pauli string on the inactive qubits containing only $\hat{\id}$ and $\hat{Z}$. Then the following hold.
\begin{itemize}
\item[(i)] The encoding gate complexity is $G_{\mathrm{enc}}(n)=O(a(n))=O(\log n)$.
\item[(ii)] The ideal score satisfies
\begin{eqnarray}
\abs{s_n(x)} \ge (\cos \phi_0)^{a(n)} = n^{-\gamma\log(1/\cos\phi_0) + O(1/\log n)}.
\label{eq:explicit_ideal_bias}
\end{eqnarray}
\item[(iii)] If the interleaved noise is local Pauli diagonal and obeys
\begin{eqnarray}
\Gamma_n(\hat{P}_n)=O(\log n),
\end{eqnarray}
then the noisy oracle reliability satisfies $\beta_0(n) \ge n^{-O(1)}$.
\item[(iv)] If the classical part of the hybrid model has sample rank polynomial in $n$, then the resulting projected hybrid learner has polynomial PAC sample complexity in $n$, $1/\varepsilon$, and $\log(1/\delta)$.
\end{itemize}
\end{theorem}

\begin{proof}---Part (i) is immediate from Eq.~(\ref{eq:explicit_encoder}): each active qubit receives one Hadamard gate and one phase gate, so the encoding cost is $2a(n)=O(\log n)$.

For part (ii), {\bf Proposition~\ref{prop:ideal_score}} applies because the Heisenberg image in Eq.~(\ref{eq:explicit_Qn}) contains only $\hat{X}$ operators on active qubits and only $\hat{\id}$ or $\hat{Z}$ on the inactive ones. Therefore,
\begin{eqnarray}
\abs{s_n(x)} = \prod_{j=1}^{a(n)}\abs{\cos \phi_j(x)} \ge (\cos\phi_0)^{a(n)}.
\end{eqnarray}
Since $a(n)=\lceil \gamma\log n\rceil$, this equals
\begin{eqnarray}
(\cos\phi_0)^{\lceil \gamma\log n\rceil} &=& \exp\big(\lceil \gamma\log n \rceil \log(\cos\phi_0)\big) \nonumber \\[2pt]
	&=&  n^{-\gamma\log(1/\cos\phi_0)+O(1/\log n)},
\end{eqnarray}
which proves Eq.~(\ref{eq:explicit_ideal_bias}).

For part (iii), {\bf Theorem~\ref{thm:pauli_general}} and {\bf Corollary~\ref{cor:pauli_general_poly}} imply that the noisy expectation is attenuated by a factor $\Lambda_n^{\mathrm{PD}}(\hat{P}_n)\ge n^{-O(1)}$ whenever $\Gamma_n(\hat{P}_n)=O(\log n)$. Defining the target label as $h^\star(x)=\sgn(s_n(x))$, the same argument as in {\bf Corollary~\ref{cor:poly_persistence}} gives
\begin{eqnarray}
\beta_0(n) \ge \Lambda_n^{\mathrm{PD}}(\hat{P}_n)\inf_x \abs{s_n(x)} \ge n^{-O(1)}.
\end{eqnarray}

For part (iv), the projected quantum readout uses only one observable, so by {\bf Theorem~\ref{thm:kernel_rank}} the sample regularized dimension increases by at most one beyond the classical sample rank. If the latter is polynomial in $n$, then the hybrid feature dimension is polynomial. Combining this with the inverse-polynomial reliability from part (iii) and {\bf Corollary~\ref{cor:poly_sample}} gives polynomial PAC sample complexity.
\end{proof}

{\bf Theorem~\ref{thm:explicit_family}} deliberately uses only logarithmically many active qubits. This keeps the ideal multiplicative score floor inverse-polynomial. The theorem also does not require the post-encoding Clifford block to be trivial; it only requires that the chosen measured observable have a backward light cone whose total attenuation budget is $O(\log n)$. Thus the overall number of qubits and the total circuit size may still grow polynomially with $n$.

\section{Toy-model details, stronger classical baselines, and numerical protocol}\label{app:toy}

This appendix records the concrete feature counts and the numerical protocol for the $64$-qubit toy family of Sec.~\ref{sec:toy}, together with the additional classical-kernel controls used in Appendix~\ref{app:kernel_controls}.

\begin{proposition}[The exact classical interaction appears at degree eight]\label{prop:deg8_exact}
Let
\begin{eqnarray}
x_{\mathrm{sel}} := \big( z_7, z_8, \cos\phi_1, \sin\phi_1, \ldots, \cos\phi_6, \sin\phi_6 \big) \in \R^{14}.
\end{eqnarray}
Then, the projected quantum feature in Eq.~(\ref{eq:q_def}) is the interaction-only degree-eight monomial
\begin{eqnarray}
q(x) = z_7 z_8 (\cos\phi_1)(\cos\phi_2)(\sin\phi_3)(\cos\phi_4)(\cos\phi_5)(\sin\phi_6).
\label{eq:q_degree8}
\end{eqnarray}
Hence, the exact degree-eight interaction-only classical feature class in $x_{\mathrm{sel}}$ contains the target interaction exactly, whereas every interaction-only expansion of degree at most $7$ necessarily misses this monomial.
\end{proposition}

\begin{proof}---Eq.~(\ref{eq:q_degree8}) is just Eq.~(\ref{eq:q_def}) rewritten in the selected raw variables. The right-hand side has total degree $8$, uses each raw variable at most once, and is therefore an interaction-only degree-eight monomial. Consequently, it belongs to the exact degree-eight interaction family and does not belong to any degree-at-most-$7$ interaction-only expansion.
\end{proof}

{\bf Proposition~\ref{prop:deg8_exact}} explains why we report, in Fig.~\ref{fig:learning_curves}(b), the size of an exact classical interaction baseline even though the main learning-curve panel focuses on lower-order baselines. The representational point is that exact classical recovery is possible only after a much larger feature expansion.

\begin{proposition}[Feature counts for the $64$-qubit experiment]\label{prop:feature_count}
The raw linear model of Sec.~\ref{sec:toy} uses $58$ context-bit features and $12$ trigonometric features, hence
\begin{eqnarray}
D_{\mathrm{lin}} = 58+12 = 70.
\end{eqnarray}
Let $d=14$ be the dimension of $x_{\mathrm{sel}}$. Then, the interaction-only degree-at-most-$3$ classical expansion has
\begin{eqnarray}
D_{\le 3}^{\mathrm{int}} = \sum_{r=1}^{3}\binom{d}{r} = \binom{14}{1}+\binom{14}{2}+\binom{14}{3} = 469
\end{eqnarray}
features, the exact degree-eight interaction-only expansion has
\begin{eqnarray}
D_{=8}^{\mathrm{int}} = \binom{14}{8} = 3003,
\end{eqnarray}
and the cumulative interaction-only family up to degree $8$ has
\begin{eqnarray}
D_{\le 8}^{\mathrm{int}} = \sum_{r=1}^{8}\binom{14}{r} = 12{,}910
\end{eqnarray}
features. Therefore, the hybrid projected model has
\begin{eqnarray}
D_{\mathrm{hyb}} = D_{\mathrm{lin}}+1 = 71
\end{eqnarray}
features.
\end{proposition}

\begin{proof}---The linear count is immediate from the definition of $\Phi_{\mathrm{lin}}$. For the interaction-only expansions, the number of degree-exactly-$r$ monomials in $d$ variables is $\binom{d}{r}$ because each monomial corresponds to choosing which $r$ distinct variables appear. Summing from $r=1$ to $3$ gives
\begin{eqnarray}
14+91+364=469.
\end{eqnarray}
Similarly,
\begin{eqnarray}
\binom{14}{8}=3003,
\end{eqnarray}
while
\begin{eqnarray}
\sum_{r=1}^{8}\binom{14}{r} &=& 14+91+364+1001+2002 \nonumber \\
	&& \quad +3003+3432+3003 = 12{,}910.
\end{eqnarray}
Adding the single projected quantum feature to the $70$ raw linear features gives the hybrid dimension $71$.
\end{proof}

The released code evaluates the projected quantum feature analytically from Eq.~(\ref{eq:q_def}), so it never allocates a full $2^{64}$-dimensional statevector. This is important for the interpretation of Sec.~\ref{sec:toy}: the numerical cost tracks the active-support description of the model rather than the full Hilbert-space dimension.

The algorithm below is exactly the protocol used to generate Fig.~\ref{fig:learning_curves}. 
\begin{figure*}[h]
\centering
\begin{minipage}{0.94\textwidth}
\hrule\vspace{0.2em}\hrule\vspace{0.6em}
\textbf{Algorithm 1. Numerical protocol for the $64$-qubit benchmark suite}\vspace{0.6em}
\hrule\vspace{0.2em}\hrule
\vspace{0.5em}
\begin{algorithmic}[1]
\Require total qubits $m=64$; active set $S=\{1,\ldots,6\}$; informative context bits $\{7,8\}$; training sizes $\mathcal N=\{80,160,320,640,1280\}$; repetitions $R=5$; clean test size $N_{\mathrm{test}}=4000$; label-noise rate $\eta=0.1$; phase window $\Delta=\pi/6$.
\For{$r=1$ to $R$}
    \State Draw a clean test set of size $N_{\mathrm{test}}$ by sampling $b\in\{0,1\}^{58}$ uniformly.
    \For{$j\in\{1,2,4,5\}$}
        \State Sample $\phi_j$ uniformly from $\{0,\pi\}+[-\Delta,\Delta]$.
    \EndFor
    \For{$j\in\{3,6\}$}
        \State Sample $\phi_j$ uniformly from $\{\pi/2,3\pi/2\}+[-\Delta,\Delta]$.
    \EndFor
    \State Compute the projected quantum feature $q(x)$ analytically from Eq.~\eqref{eq:q_def}.
    \State Set the clean label $y=\sgn(q(x))$.
    \For{each $N\in\mathcal N$}
        \State Draw a training set of size $N$ from the same distribution.
        \State Flip each training label independently with probability $\eta$.
        \State Build the four explicit feature families: linear raw, cubic selected, exact degree-$8$ selected, and hybrid projected.
        \State Fit ridge-type linear classifiers to the noisy training set.
        \State Evaluate clean-test accuracy on the fixed test set for repetition $r$.
    \EndFor
\EndFor
\State Report the mean and standard deviation across the $R$ repetitions.
\end{algorithmic}
\vspace{0.6em}\hrule\vspace{0.2em}\hrule
\end{minipage}
\end{figure*}

\subsection{Additional classical-kernel controls}\label{app:kernel_controls} 

\begin{figure*}[t]
\centering
\includegraphics[width=0.60\textwidth]{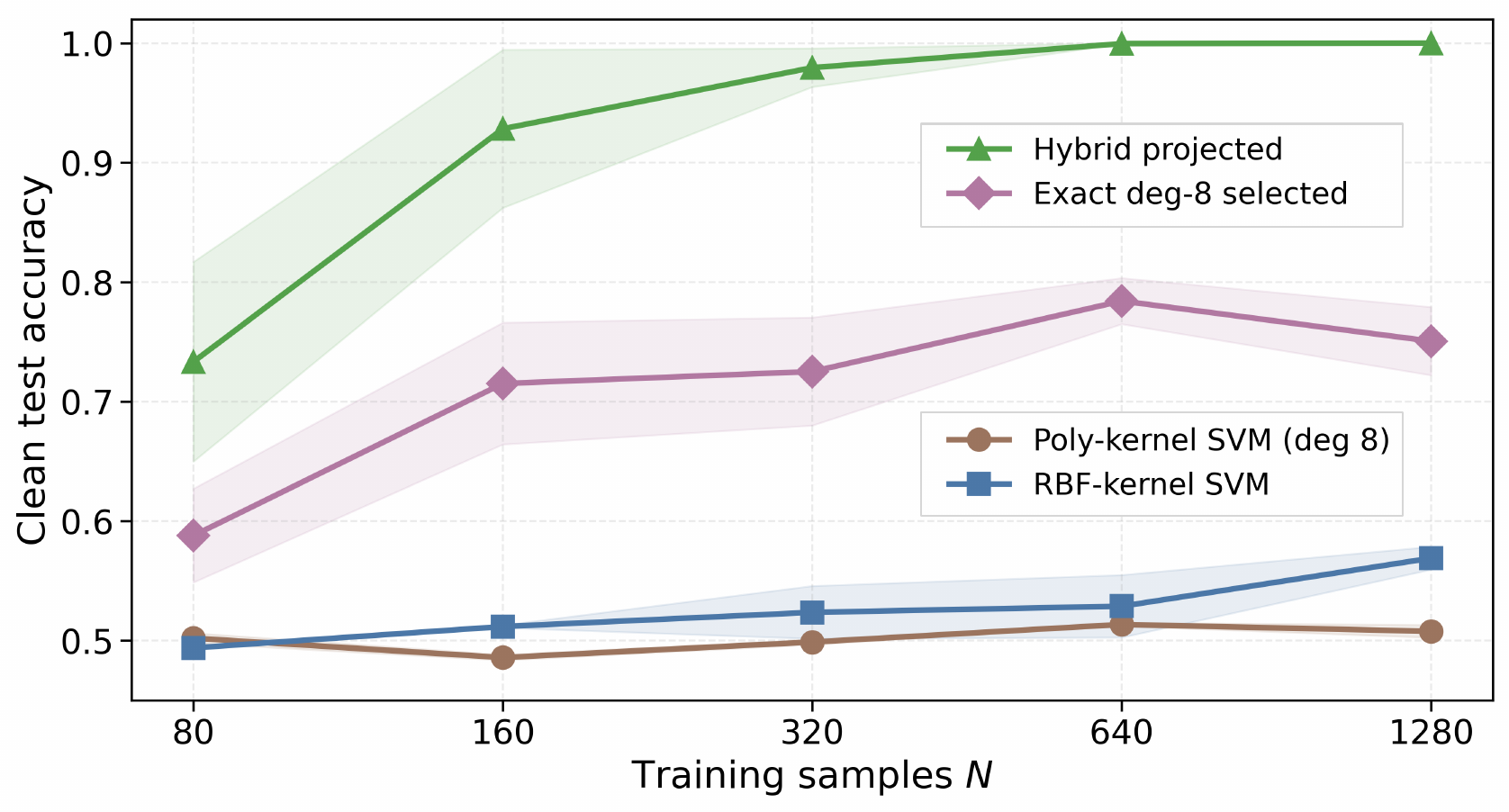}
\caption{Additional classical-kernel controls on the $64$-qubit benchmark. The exact degree-eight selected baseline from Fig.~\ref{fig:learning_curves} is shown again for reference. On this task, the pilot-tuned degree-eight polynomial-kernel and RBF-kernel SVM controls remain close to chance across the tested range, whereas the hybrid projected model remains strongest.}
\label{fig:kernel_controls}
\end{figure*}

To probe whether implicit classical kernels can recover the same interaction more efficiently, we also fit degree-eight polynomial-kernel and RBF-kernel support-vector machines~\cite{CortesVapnik1995} on $x_{\mathrm{sel}}$. Their hyperparameters are selected once on a pilot validation split and then fixed across the learning-curve runs. In the released experiments, the pilot search chooses $(C,\gamma,\mathrm{coef0})=(0.5,0.05,1.0)$ for the polynomial kernel and $(C,\gamma)=(2.0,1.0)$ for the RBF kernel. Fig.~\ref{fig:kernel_controls} reports the resulting clean-test accuracies.

\section{Projected kernels versus naive global kernels}\label{app:global_kernel} 

To make the projected-versus-global-kernel comparison fully explicit, we consider a smaller $6$-qubit analogue with four active phases and two context bits. The projected observable is
\begin{eqnarray}
q_{\mathrm{small}}(x)=(-1)^{b_1+b_2}\cos\phi_1\cos\phi_2\sin\phi_3\sin\phi_4,
\end{eqnarray}
and the clean label is $\sgn(q_{\mathrm{small}}(x))$. The projected AQSE kernel is built from a feature family of size $M$, while the naive full-data quantumization baseline uses the global fidelity kernel
\begin{eqnarray}
K_{\mathrm{fid}}\big((b,\phi),(b',\phi')\big)
=
\delta_{b,b'}\prod_{j=1}^{4}\cos^2\!\left(\frac{\phi_j-\phi_j'}{2}\right).
\end{eqnarray}
This is exactly the kernel obtained by comparing the full product states of the active phase qubits together with the basis-encoded context bits.

\begin{figure*}[t]
\centering
\includegraphics[width=0.94\textwidth]{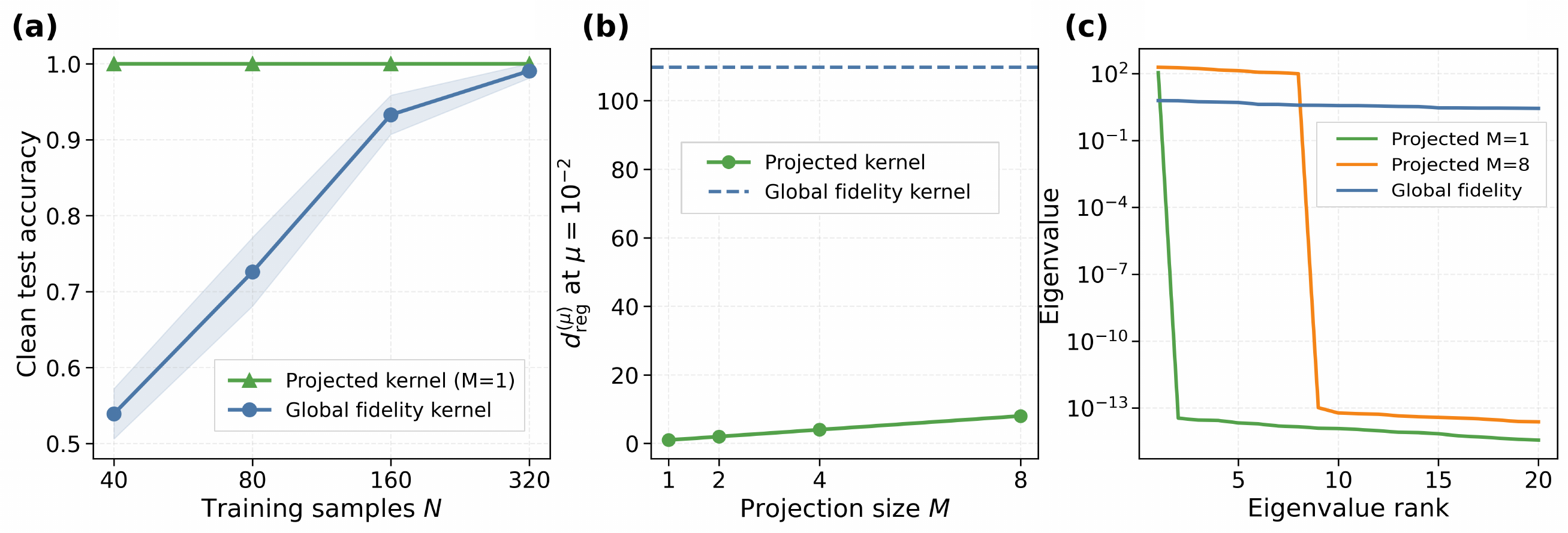}
\caption{Projected AQSE kernels versus a naive global fidelity kernel on a fully computable small-qubit analogue. Panel (a) shows clean-test accuracy versus the number of training samples $N$. The projected kernel with $M=1$ already reaches perfect accuracy, while the global fidelity kernel improves only gradually with sample size. Panel (b) shows the sample regularized dimension $d_{\mathrm{reg}}^{(\mu)}$ at $\mu=10^{-2}$: the projected kernel tracks the projection size $M$, whereas the global fidelity kernel remains of order the sample size. Panel (c) shows representative eigenvalue spectra at $N=160$. The projected kernels concentrate spectral weight on only a few modes, whereas the global fidelity kernel remains broad.}
\label{fig:global_kernel}
\end{figure*}

The numerical message matches {\bf Theorem~\ref{thm:kernel_rank}} and is a finite-sample analogue of recent warnings for global fidelity-type quantum kernels~\cite{Thanasilp2024,Agliardi2026}. For the projected families used here, $d_{\mathrm{reg}}^{(\mu)}$ is numerically $1.00$, $2.00$, $4.00$, and $8.00$ for $M=1,2,4,8$, respectively. By contrast, the global fidelity kernel has $d_{\mathrm{reg}}^{(\mu)}\approx 110$ on a sample of size $120$. The projected kernel therefore achieves the desired prediction with far fewer effective degrees of freedom, while the global fidelity kernel only catches up after many more samples.

\section{Residual-screening protocol for candidate active subspaces}\label{app:residual_protocol}

This appendix records the concrete regression family used in Fig.~\ref{fig:screening_residual}. We generate an eight-generator family
\begin{eqnarray}
g(x)=\big(z_1,z_2,\cos\theta_1,\cos\theta_2,\sin\theta_3,\cos\theta_4,\cos\theta_5,\sin\theta_6\big),
\end{eqnarray}
with unbiased context bits $z_1,z_2\in\{\pm 1\}$, independent Gaussian classical covariates $x_C\in\R^6$, and independent phases $\theta_j\sim \mathrm{Unif}[0,2\pi)$. The classical design matrix is
\begin{eqnarray}
F_C = \big[\mathbf{1},x_C,g(x)\big],
\end{eqnarray}
and the target is
\begin{eqnarray}
y = F_C w_C + \alpha \, g_1 g_2 g_3 g_5 + \xi,
\end{eqnarray}
where $w_C$ is fixed, $\alpha\in\{0.6,1.0,1.4\}$, and $\xi\sim\mathcal N(0,0.2^2)$. We screen all $\binom{8}{4}=70$ candidate $4$-subsets, compute the score
\begin{eqnarray}
\hat{s}(S)=\frac{1}{N}\big\|P_{Q_\perp(S)}r_C^{(N)}\big\|_2^2,
\end{eqnarray}
on the training split of size $N=250$, and evaluate the resulting test-risk gain on an independent test split of size $4000$. The mean Pearson correlation between $\hat{s}(S)$ and the test-risk gain is $0.984$, and the top-scoring candidate coincides with the exact active support in every run of the synthetic family. This is the empirical protocol behind the claim made in Sec.~\ref{sec:toy}.

\section{Finite-shot projected readout under dephasing and readout noise}\label{app:finite_shot}

\begin{figure*}[t]
\centering
\includegraphics[width=0.68\textwidth]{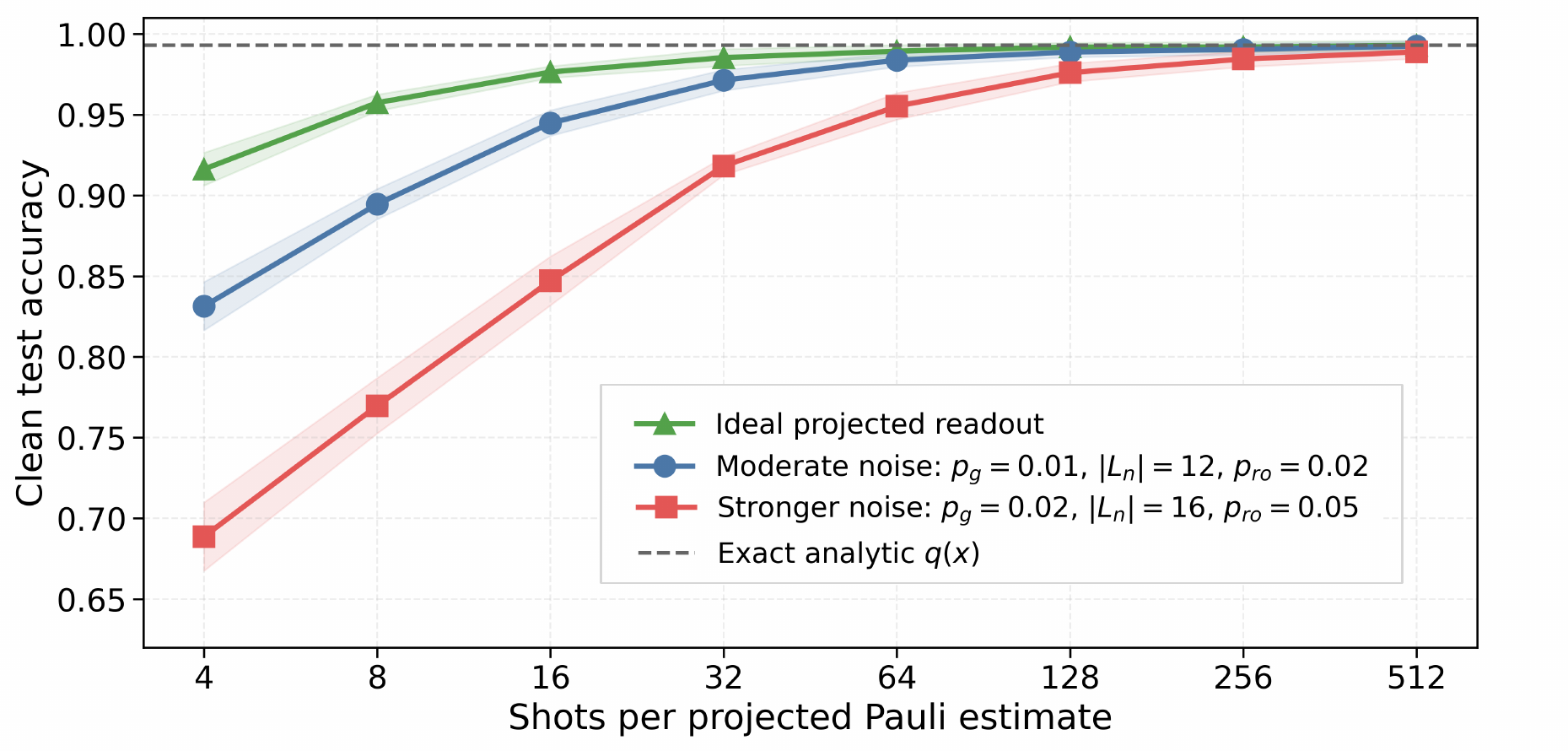}
\caption{Finite-shot projected readout for the $64$-qubit toy family. The dashed horizontal line is the exact analytic-$q(x)$ hybrid baseline. Even with combined dephasing and readout noise, moderate shot budgets recover most of the benefit: for the moderate-noise setting with attenuation $A\approx 0.753$, $32$ shots already yield about $0.971$ clean-test accuracy, and $128$ shots yield about $0.989$; for the stronger-noise setting with $A\approx 0.468$, the same accuracies are about $0.918$ and $0.976$, respectively.}
\label{fig:finite_shot}
\end{figure*}

To complement the exact-score calculations of Sec.~\ref{sec:scalable}, we simulate finite-shot estimation of the projected Pauli observable on the $64$-qubit toy family. If the noiseless projected feature is $q(x)$, then under local dephasing with per-location rate $p_g$, backward-light-cone size $|L_n|$, and readout error rate $p_{\mathrm{ro}}$, the measurement mean is attenuated to
\begin{eqnarray}
\mathbb{E}[M(x)\,|\,x] = A\,q(x), \quad A = (1-2p_g)^{|L_n|}(1-2p_{\mathrm{ro}}).
\end{eqnarray}
For a shot budget $S$, we estimate the projected quantum feature by
\begin{eqnarray}
\hat{q}_S(x)=\frac{1}{S}\sum_{t=1}^{S} M_t(x),
\quad
M_t(x)\in\{-1,+1\},
\end{eqnarray}
and train the same hybrid classifier as in Sec.~\ref{sec:toy} on $\Phi_{\mathrm{lin}}\oplus \hat{q}_S$.

The finite-shot data reinforce the main narrative of the paper. The projected readout remains operational with modest shot budgets, and the penalty from moderate device noise appears primarily as a predictable attenuation of the effective score, in line with {\bf Theorem~\ref{thm:dephasing}} and {\bf Corollary~\ref{cor:poly_persistence}}.

\twocolumngrid
\bibliographystyle{apsrev4-2}

\begin{thebibliography}{33}%
\makeatletter
\providecommand \@ifxundefined [1]{%
 \@ifx{#1\undefined}
}%
\providecommand \@ifnum [1]{%
 \ifnum #1\expandafter \@firstoftwo
 \else \expandafter \@secondoftwo
 \fi
}%
\providecommand \@ifx [1]{%
 \ifx #1\expandafter \@firstoftwo
 \else \expandafter \@secondoftwo
 \fi
}%
\providecommand \natexlab [1]{#1}%
\providecommand \enquote  [1]{``#1''}%
\providecommand \bibnamefont  [1]{#1}%
\providecommand \bibfnamefont [1]{#1}%
\providecommand \citenamefont [1]{#1}%
\providecommand \href@noop [0]{\@secondoftwo}%
\providecommand \href [0]{\begingroup \@sanitize@url \@href}%
\providecommand \@href[1]{\@@startlink{#1}\@@href}%
\providecommand \@@href[1]{\endgroup#1\@@endlink}%
\providecommand \@sanitize@url [0]{\catcode `\\12\catcode `\$12\catcode
  `\&12\catcode `\#12\catcode `\^12\catcode `\_12\catcode `\%12\relax}%
\providecommand \@@startlink[1]{}%
\providecommand \@@endlink[0]{}%
\providecommand \url  [0]{\begingroup\@sanitize@url \@url }%
\providecommand \@url [1]{\endgroup\@href {#1}{\urlprefix }}%
\providecommand \urlprefix  [0]{URL }%
\providecommand \Eprint [0]{\href }%
\providecommand \doibase [0]{https://doi.org/}%
\providecommand \selectlanguage [0]{\@gobble}%
\providecommand \bibinfo  [0]{\@secondoftwo}%
\providecommand \bibfield  [0]{\@secondoftwo}%
\providecommand \translation [1]{[#1]}%
\providecommand \BibitemOpen [0]{}%
\providecommand \bibitemStop [0]{}%
\providecommand \bibitemNoStop [0]{.\EOS\space}%
\providecommand \EOS [0]{\spacefactor3000\relax}%
\providecommand \BibitemShut  [1]{\csname bibitem#1\endcsname}%
\let\auto@bib@innerbib\@empty
\bibitem [{\citenamefont {Biamonte}\ \emph {et~al.}(2017)\citenamefont
  {Biamonte}, \citenamefont {Wittek}, \citenamefont {Pancotti}, \citenamefont
  {Rebentrost}, \citenamefont {Wiebe},\ and\ \citenamefont
  {Lloyd}}]{Biamonte2017}%
  \BibitemOpen
  \bibfield  {author} {\bibinfo {author} {\bibfnamefont {J.}~\bibnamefont
  {Biamonte}}, \bibinfo {author} {\bibfnamefont {P.}~\bibnamefont {Wittek}},
  \bibinfo {author} {\bibfnamefont {N.}~\bibnamefont {Pancotti}}, \bibinfo
  {author} {\bibfnamefont {P.}~\bibnamefont {Rebentrost}}, \bibinfo {author}
  {\bibfnamefont {N.}~\bibnamefont {Wiebe}},\ and\ \bibinfo {author}
  {\bibfnamefont {S.}~\bibnamefont {Lloyd}},\ }\href
  {https://doi.org/10.1038/nature23474} {\bibfield  {journal} {\bibinfo
  {journal} {Nature}\ }\textbf {\bibinfo {volume} {549}},\ \bibinfo {pages}
  {195} (\bibinfo {year} {2017})}\BibitemShut {NoStop}%
\bibitem [{\citenamefont {Ciliberto}\ \emph {et~al.}(2018)\citenamefont
  {Ciliberto}, \citenamefont {Herbster}, \citenamefont {Ialongo}, \citenamefont
  {Pontil}, \citenamefont {Rocchetto}, \citenamefont {Severini},\ and\
  \citenamefont {Wossnig}}]{Ciliberto2018}%
  \BibitemOpen
  \bibfield  {author} {\bibinfo {author} {\bibfnamefont {C.}~\bibnamefont
  {Ciliberto}}, \bibinfo {author} {\bibfnamefont {M.}~\bibnamefont {Herbster}},
  \bibinfo {author} {\bibfnamefont {A.~D.}\ \bibnamefont {Ialongo}}, \bibinfo
  {author} {\bibfnamefont {M.}~\bibnamefont {Pontil}}, \bibinfo {author}
  {\bibfnamefont {A.}~\bibnamefont {Rocchetto}}, \bibinfo {author}
  {\bibfnamefont {S.}~\bibnamefont {Severini}},\ and\ \bibinfo {author}
  {\bibfnamefont {L.}~\bibnamefont {Wossnig}},\ }\href
  {https://doi.org/10.1098/rspa.2017.0551} {\bibfield  {journal} {\bibinfo
  {journal} {Proceedings of the Royal Society A: Mathematical, Physical and
  Engineering Sciences}\ }\textbf {\bibinfo {volume} {474}},\ \bibinfo {pages}
  {20170551} (\bibinfo {year} {2018})}\BibitemShut {NoStop}%
\bibitem [{\citenamefont {Havl{\'i}{\v{c}}ek}\ \emph
  {et~al.}(2019)\citenamefont {Havl{\'i}{\v{c}}ek}, \citenamefont
  {C{\'o}rcoles}, \citenamefont {Temme}, \citenamefont {Harrow}, \citenamefont
  {Kandala}, \citenamefont {Chow},\ and\ \citenamefont
  {Gambetta}}]{Havlicek2019}%
  \BibitemOpen
  \bibfield  {author} {\bibinfo {author} {\bibfnamefont {V.}~\bibnamefont
  {Havl{\'i}{\v{c}}ek}}, \bibinfo {author} {\bibfnamefont {A.~D.}\ \bibnamefont
  {C{\'o}rcoles}}, \bibinfo {author} {\bibfnamefont {K.}~\bibnamefont {Temme}},
  \bibinfo {author} {\bibfnamefont {A.~W.}\ \bibnamefont {Harrow}}, \bibinfo
  {author} {\bibfnamefont {A.}~\bibnamefont {Kandala}}, \bibinfo {author}
  {\bibfnamefont {J.~M.}\ \bibnamefont {Chow}},\ and\ \bibinfo {author}
  {\bibfnamefont {J.~M.}\ \bibnamefont {Gambetta}},\ }\href
  {https://doi.org/10.1038/s41586-019-0980-2} {\bibfield  {journal} {\bibinfo
  {journal} {Nature}\ }\textbf {\bibinfo {volume} {567}},\ \bibinfo {pages}
  {209} (\bibinfo {year} {2019})}\BibitemShut {NoStop}%
\bibitem [{\citenamefont {Schuld}\ and\ \citenamefont
  {Killoran}(2019)}]{SchuldKilloran2019}%
  \BibitemOpen
  \bibfield  {author} {\bibinfo {author} {\bibfnamefont {M.}~\bibnamefont
  {Schuld}}\ and\ \bibinfo {author} {\bibfnamefont {N.}~\bibnamefont
  {Killoran}},\ }\href {https://doi.org/10.1103/PhysRevLett.122.040504}
  {\bibfield  {journal} {\bibinfo  {journal} {Physical Review Letters}\
  }\textbf {\bibinfo {volume} {122}},\ \bibinfo {pages} {040504} (\bibinfo
  {year} {2019})}\BibitemShut {NoStop}%
\bibitem [{\citenamefont {Cerezo}\ \emph {et~al.}(2022)\citenamefont {Cerezo},
  \citenamefont {Verdon}, \citenamefont {Huang}, \citenamefont {Cincio},\ and\
  \citenamefont {Coles}}]{Cerezo2022}%
  \BibitemOpen
  \bibfield  {author} {\bibinfo {author} {\bibfnamefont {M.}~\bibnamefont
  {Cerezo}}, \bibinfo {author} {\bibfnamefont {G.}~\bibnamefont {Verdon}},
  \bibinfo {author} {\bibfnamefont {H.-Y.}\ \bibnamefont {Huang}}, \bibinfo
  {author} {\bibfnamefont {{\L}.}~\bibnamefont {Cincio}},\ and\ \bibinfo
  {author} {\bibfnamefont {P.~J.}\ \bibnamefont {Coles}},\ }\href
  {https://doi.org/10.1038/s43588-022-00311-3} {\bibfield  {journal} {\bibinfo
  {journal} {Nature Computational Science}\ }\textbf {\bibinfo {volume} {2}},\
  \bibinfo {pages} {567} (\bibinfo {year} {2022})}\BibitemShut {NoStop}%
\bibitem [{\citenamefont {Liu}\ \emph {et~al.}(2021)\citenamefont {Liu},
  \citenamefont {Arunachalam},\ and\ \citenamefont {Temme}}]{Liu2021NatPhys}%
  \BibitemOpen
  \bibfield  {author} {\bibinfo {author} {\bibfnamefont {Y.}~\bibnamefont
  {Liu}}, \bibinfo {author} {\bibfnamefont {S.}~\bibnamefont {Arunachalam}},\
  and\ \bibinfo {author} {\bibfnamefont {K.}~\bibnamefont {Temme}},\ }\href
  {https://doi.org/10.1038/s41567-021-01287-z} {\bibfield  {journal} {\bibinfo
  {journal} {Nature Physics}\ }\textbf {\bibinfo {volume} {17}},\ \bibinfo
  {pages} {1013} (\bibinfo {year} {2021})}\BibitemShut {NoStop}%
\bibitem [{\citenamefont {Preskill}(2018)}]{Preskill2018}%
  \BibitemOpen
  \bibfield  {author} {\bibinfo {author} {\bibfnamefont {J.}~\bibnamefont
  {Preskill}},\ }\href {https://doi.org/10.22331/q-2018-08-06-79} {\bibfield
  {journal} {\bibinfo  {journal} {Quantum}\ }\textbf {\bibinfo {volume} {2}},\
  \bibinfo {pages} {79} (\bibinfo {year} {2018})}\BibitemShut {NoStop}%
\bibitem [{\citenamefont {Schuld}\ and\ \citenamefont
  {Killoran}(2022)}]{Schuld2022PRXQuantum}%
  \BibitemOpen
  \bibfield  {author} {\bibinfo {author} {\bibfnamefont {M.}~\bibnamefont
  {Schuld}}\ and\ \bibinfo {author} {\bibfnamefont {N.}~\bibnamefont
  {Killoran}},\ }\href {https://doi.org/10.1103/PRXQuantum.3.030101} {\bibfield
   {journal} {\bibinfo  {journal} {PRX Quantum}\ }\textbf {\bibinfo {volume}
  {3}},\ \bibinfo {pages} {030101} (\bibinfo {year} {2022})}\BibitemShut
  {NoStop}%
\bibitem [{\citenamefont {Aaronson}(2015)}]{Aaronson2015}%
  \BibitemOpen
  \bibfield  {author} {\bibinfo {author} {\bibfnamefont {S.}~\bibnamefont
  {Aaronson}},\ }\href {https://doi.org/10.1038/nphys3272} {\bibfield
  {journal} {\bibinfo  {journal} {Nature Physics}\ }\textbf {\bibinfo {volume}
  {11}},\ \bibinfo {pages} {291} (\bibinfo {year} {2015})}\BibitemShut
  {NoStop}%
\bibitem [{\citenamefont {Paler}\ \emph {et~al.}(2020)\citenamefont {Paler},
  \citenamefont {Oumarou},\ and\ \citenamefont {Basmadjian}}]{Paler2020}%
  \BibitemOpen
  \bibfield  {author} {\bibinfo {author} {\bibfnamefont {A.}~\bibnamefont
  {Paler}}, \bibinfo {author} {\bibfnamefont {O.}~\bibnamefont {Oumarou}},\
  and\ \bibinfo {author} {\bibfnamefont {R.}~\bibnamefont {Basmadjian}},\
  }\href {https://doi.org/10.1103/PhysRevA.102.032608} {\bibfield  {journal}
  {\bibinfo  {journal} {Physical Review A}\ }\textbf {\bibinfo {volume}
  {102}},\ \bibinfo {pages} {032608} (\bibinfo {year} {2020})}\BibitemShut
  {NoStop}%
\bibitem [{\citenamefont {Harrow}(2020)}]{Harrow2020}%
  \BibitemOpen
  \bibfield  {author} {\bibinfo {author} {\bibfnamefont {A.~W.}\ \bibnamefont
  {Harrow}},\ }\href {https://doi.org/10.48550/arXiv.2004.00026} {\bibinfo
  {title} {Small quantum computers and large classical data sets}} (\bibinfo
  {year} {2020}),\ \Eprint {https://arxiv.org/abs/2004.00026} {arXiv:2004.00026
  [quant-ph]} \BibitemShut {NoStop}%
\bibitem [{\citenamefont {Tang}(2021)}]{Tang2021}%
  \BibitemOpen
  \bibfield  {author} {\bibinfo {author} {\bibfnamefont {E.}~\bibnamefont
  {Tang}},\ }\href {https://doi.org/10.1103/PhysRevLett.127.060503} {\bibfield
  {journal} {\bibinfo  {journal} {Physical Review Letters}\ }\textbf {\bibinfo
  {volume} {127}},\ \bibinfo {pages} {060503} (\bibinfo {year}
  {2021})}\BibitemShut {NoStop}%
\bibitem [{\citenamefont {Jerbi}\ \emph {et~al.}(2023)\citenamefont {Jerbi},
  \citenamefont {Fiderer}, \citenamefont {Poulsen~Nautrup}, \citenamefont
  {K{\"u}bler}, \citenamefont {Briegel},\ and\ \citenamefont
  {Dunjko}}]{Jerbi2023}%
  \BibitemOpen
  \bibfield  {author} {\bibinfo {author} {\bibfnamefont {S.}~\bibnamefont
  {Jerbi}}, \bibinfo {author} {\bibfnamefont {L.~J.}\ \bibnamefont {Fiderer}},
  \bibinfo {author} {\bibfnamefont {H.}~\bibnamefont {Poulsen~Nautrup}},
  \bibinfo {author} {\bibfnamefont {J.~M.}\ \bibnamefont {K{\"u}bler}},
  \bibinfo {author} {\bibfnamefont {H.~J.}\ \bibnamefont {Briegel}},\ and\
  \bibinfo {author} {\bibfnamefont {V.}~\bibnamefont {Dunjko}},\ }\href
  {https://doi.org/10.1038/s41467-023-36159-y} {\bibfield  {journal} {\bibinfo
  {journal} {Nature Communications}\ }\textbf {\bibinfo {volume} {14}},\
  \bibinfo {pages} {517} (\bibinfo {year} {2023})}\BibitemShut {NoStop}%
\bibitem [{\citenamefont {Huang}\ \emph {et~al.}(2021)\citenamefont {Huang},
  \citenamefont {Broughton}, \citenamefont {Mohseni}, \citenamefont {Babbush},
  \citenamefont {Boixo}, \citenamefont {Neven},\ and\ \citenamefont
  {McClean}}]{Huang2021}%
  \BibitemOpen
  \bibfield  {author} {\bibinfo {author} {\bibfnamefont {H.-Y.}\ \bibnamefont
  {Huang}}, \bibinfo {author} {\bibfnamefont {M.}~\bibnamefont {Broughton}},
  \bibinfo {author} {\bibfnamefont {M.}~\bibnamefont {Mohseni}}, \bibinfo
  {author} {\bibfnamefont {R.}~\bibnamefont {Babbush}}, \bibinfo {author}
  {\bibfnamefont {S.}~\bibnamefont {Boixo}}, \bibinfo {author} {\bibfnamefont
  {H.}~\bibnamefont {Neven}},\ and\ \bibinfo {author} {\bibfnamefont {J.~R.}\
  \bibnamefont {McClean}},\ }\href {https://doi.org/10.1038/s41467-021-22539-9}
  {\bibfield  {journal} {\bibinfo  {journal} {Nature Communications}\ }\textbf
  {\bibinfo {volume} {12}},\ \bibinfo {pages} {2631} (\bibinfo {year}
  {2021})}\BibitemShut {NoStop}%
\bibitem [{\citenamefont {Thanasilp}\ \emph {et~al.}(2024)\citenamefont
  {Thanasilp}, \citenamefont {Wang}, \citenamefont {Cerezo},\ and\
  \citenamefont {Holmes}}]{Thanasilp2024}%
  \BibitemOpen
  \bibfield  {author} {\bibinfo {author} {\bibfnamefont {S.}~\bibnamefont
  {Thanasilp}}, \bibinfo {author} {\bibfnamefont {S.}~\bibnamefont {Wang}},
  \bibinfo {author} {\bibfnamefont {M.}~\bibnamefont {Cerezo}},\ and\ \bibinfo
  {author} {\bibfnamefont {Z.}~\bibnamefont {Holmes}},\ }\href
  {https://doi.org/10.1038/s41467-024-49287-w} {\bibfield  {journal} {\bibinfo
  {journal} {Nature Communications}\ }\textbf {\bibinfo {volume} {15}},\
  \bibinfo {pages} {5200} (\bibinfo {year} {2024})}\BibitemShut {NoStop}%
\bibitem [{\citenamefont {Agliardi}\ \emph {et~al.}(2026)\citenamefont
  {Agliardi}, \citenamefont {Cortiana}, \citenamefont {Dekusar}, \citenamefont
  {Ghosh}, \citenamefont {Mohseni}, \citenamefont {O'Meara}, \citenamefont
  {Valls}, \citenamefont {Yogaraj},\ and\ \citenamefont {Zhuk}}]{Agliardi2026}%
  \BibitemOpen
  \bibfield  {author} {\bibinfo {author} {\bibfnamefont {G.}~\bibnamefont
  {Agliardi}}, \bibinfo {author} {\bibfnamefont {G.}~\bibnamefont {Cortiana}},
  \bibinfo {author} {\bibfnamefont {A.}~\bibnamefont {Dekusar}}, \bibinfo
  {author} {\bibfnamefont {K.}~\bibnamefont {Ghosh}}, \bibinfo {author}
  {\bibfnamefont {N.}~\bibnamefont {Mohseni}}, \bibinfo {author} {\bibfnamefont
  {C.}~\bibnamefont {O'Meara}}, \bibinfo {author} {\bibfnamefont
  {V.}~\bibnamefont {Valls}}, \bibinfo {author} {\bibfnamefont
  {K.}~\bibnamefont {Yogaraj}},\ and\ \bibinfo {author} {\bibfnamefont
  {S.}~\bibnamefont {Zhuk}},\ }\href
  {https://doi.org/10.1038/s41534-025-01154-2} {\bibfield  {journal} {\bibinfo
  {journal} {npj Quantum Information}\ }\textbf {\bibinfo {volume} {12}},\
  \bibinfo {pages} {12} (\bibinfo {year} {2026})}\BibitemShut {NoStop}%
\bibitem [{\citenamefont {Song}\ \emph {et~al.}(2021)\citenamefont {Song},
  \citenamefont {Wie{\'s}niak}, \citenamefont {Liu}, \citenamefont
  {Paw{\l}owski}, \citenamefont {Lee}, \citenamefont {Kim},\ and\ \citenamefont
  {Bang}}]{Song2021}%
  \BibitemOpen
  \bibfield  {author} {\bibinfo {author} {\bibfnamefont {W.}~\bibnamefont
  {Song}}, \bibinfo {author} {\bibfnamefont {M.}~\bibnamefont {Wie{\'s}niak}},
  \bibinfo {author} {\bibfnamefont {N.}~\bibnamefont {Liu}}, \bibinfo {author}
  {\bibfnamefont {M.}~\bibnamefont {Paw{\l}owski}}, \bibinfo {author}
  {\bibfnamefont {J.}~\bibnamefont {Lee}}, \bibinfo {author} {\bibfnamefont
  {J.}~\bibnamefont {Kim}},\ and\ \bibinfo {author} {\bibfnamefont
  {J.}~\bibnamefont {Bang}},\ }\href
  {https://doi.org/10.1007/s11128-021-03217-7} {\bibfield  {journal} {\bibinfo
  {journal} {Quantum Information Processing}\ }\textbf {\bibinfo {volume}
  {20}},\ \bibinfo {pages} {275} (\bibinfo {year} {2021})}\BibitemShut
  {NoStop}%
\bibitem [{\citenamefont {Song}\ \emph {et~al.}(2022)\citenamefont {Song},
  \citenamefont {Lim}, \citenamefont {Jeong}, \citenamefont {Ji}, \citenamefont
  {Lee}, \citenamefont {Kim}, \citenamefont {Kim},\ and\ \citenamefont
  {Bang}}]{Song2022}%
  \BibitemOpen
  \bibfield  {author} {\bibinfo {author} {\bibfnamefont {W.}~\bibnamefont
  {Song}}, \bibinfo {author} {\bibfnamefont {Y.}~\bibnamefont {Lim}}, \bibinfo
  {author} {\bibfnamefont {K.}~\bibnamefont {Jeong}}, \bibinfo {author}
  {\bibfnamefont {Y.-S.}\ \bibnamefont {Ji}}, \bibinfo {author} {\bibfnamefont
  {J.}~\bibnamefont {Lee}}, \bibinfo {author} {\bibfnamefont {J.}~\bibnamefont
  {Kim}}, \bibinfo {author} {\bibfnamefont {M.~S.}\ \bibnamefont {Kim}},\ and\
  \bibinfo {author} {\bibfnamefont {J.}~\bibnamefont {Bang}},\ }\href
  {https://doi.org/10.1088/2058-9565/ac51b0} {\bibfield  {journal} {\bibinfo
  {journal} {Quantum Science and Technology}\ }\textbf {\bibinfo {volume}
  {7}},\ \bibinfo {pages} {025009} (\bibinfo {year} {2022})}\BibitemShut
  {NoStop}%
\bibitem [{\citenamefont {Caro}\ \emph {et~al.}(2022)\citenamefont {Caro},
  \citenamefont {Huang}, \citenamefont {Cerezo}, \citenamefont {Sharma},
  \citenamefont {Sornborger}, \citenamefont {Cincio},\ and\ \citenamefont
  {Coles}}]{Caro2022}%
  \BibitemOpen
  \bibfield  {author} {\bibinfo {author} {\bibfnamefont {M.~C.}\ \bibnamefont
  {Caro}}, \bibinfo {author} {\bibfnamefont {H.-Y.}\ \bibnamefont {Huang}},
  \bibinfo {author} {\bibfnamefont {M.}~\bibnamefont {Cerezo}}, \bibinfo
  {author} {\bibfnamefont {K.}~\bibnamefont {Sharma}}, \bibinfo {author}
  {\bibfnamefont {A.~T.}\ \bibnamefont {Sornborger}}, \bibinfo {author}
  {\bibfnamefont {{\L}.}~\bibnamefont {Cincio}},\ and\ \bibinfo {author}
  {\bibfnamefont {P.~J.}\ \bibnamefont {Coles}},\ }\href
  {https://doi.org/10.1038/s41467-022-32550-3} {\bibfield  {journal} {\bibinfo
  {journal} {Nature Communications}\ }\textbf {\bibinfo {volume} {13}},\
  \bibinfo {pages} {4919} (\bibinfo {year} {2022})}\BibitemShut {NoStop}%
\bibitem [{\citenamefont {Yin}\ \emph {et~al.}(2025)\citenamefont {Yin},
  \citenamefont {Agresti}, \citenamefont {de~Felice}, \citenamefont {Brown},
  \citenamefont {Toumi}, \citenamefont {Pentangelo}, \citenamefont
  {Piacentini}, \citenamefont {Crespi}, \citenamefont {Ceccarelli},
  \citenamefont {Osellame}, \citenamefont {Coecke},\ and\ \citenamefont
  {Walther}}]{Yin2025}%
  \BibitemOpen
  \bibfield  {author} {\bibinfo {author} {\bibfnamefont {Z.}~\bibnamefont
  {Yin}}, \bibinfo {author} {\bibfnamefont {I.}~\bibnamefont {Agresti}},
  \bibinfo {author} {\bibfnamefont {G.}~\bibnamefont {de~Felice}}, \bibinfo
  {author} {\bibfnamefont {D.}~\bibnamefont {Brown}}, \bibinfo {author}
  {\bibfnamefont {A.}~\bibnamefont {Toumi}}, \bibinfo {author} {\bibfnamefont
  {C.}~\bibnamefont {Pentangelo}}, \bibinfo {author} {\bibfnamefont
  {S.}~\bibnamefont {Piacentini}}, \bibinfo {author} {\bibfnamefont
  {A.}~\bibnamefont {Crespi}}, \bibinfo {author} {\bibfnamefont
  {F.}~\bibnamefont {Ceccarelli}}, \bibinfo {author} {\bibfnamefont
  {R.}~\bibnamefont {Osellame}}, \bibinfo {author} {\bibfnamefont
  {B.}~\bibnamefont {Coecke}},\ and\ \bibinfo {author} {\bibfnamefont
  {P.}~\bibnamefont {Walther}},\ }\href
  {https://doi.org/10.1038/s41566-025-01682-5} {\bibfield  {journal} {\bibinfo
  {journal} {Nature Photonics}\ }\textbf {\bibinfo {volume} {19}},\ \bibinfo
  {pages} {1020} (\bibinfo {year} {2025})}\BibitemShut {NoStop}%
\bibitem [{\citenamefont {Abbas}\ \emph {et~al.}(2021)\citenamefont {Abbas},
  \citenamefont {Sutter}, \citenamefont {Zoufal}, \citenamefont {Lucchi},
  \citenamefont {Figalli},\ and\ \citenamefont {Woerner}}]{Abbas2021}%
  \BibitemOpen
  \bibfield  {author} {\bibinfo {author} {\bibfnamefont {A.}~\bibnamefont
  {Abbas}}, \bibinfo {author} {\bibfnamefont {D.}~\bibnamefont {Sutter}},
  \bibinfo {author} {\bibfnamefont {C.}~\bibnamefont {Zoufal}}, \bibinfo
  {author} {\bibfnamefont {A.}~\bibnamefont {Lucchi}}, \bibinfo {author}
  {\bibfnamefont {A.}~\bibnamefont {Figalli}},\ and\ \bibinfo {author}
  {\bibfnamefont {S.}~\bibnamefont {Woerner}},\ }\href
  {https://doi.org/10.1038/s43588-021-00084-1} {\bibfield  {journal} {\bibinfo
  {journal} {Nature Computational Science}\ }\textbf {\bibinfo {volume} {1}},\
  \bibinfo {pages} {403} (\bibinfo {year} {2021})}\BibitemShut {NoStop}%
\bibitem [{\citenamefont {Huang}\ \emph {et~al.}(2020)\citenamefont {Huang},
  \citenamefont {Kueng},\ and\ \citenamefont
  {Preskill}}]{HuangKuengPreskill2020}%
  \BibitemOpen
  \bibfield  {author} {\bibinfo {author} {\bibfnamefont {H.-Y.}\ \bibnamefont
  {Huang}}, \bibinfo {author} {\bibfnamefont {R.}~\bibnamefont {Kueng}},\ and\
  \bibinfo {author} {\bibfnamefont {J.}~\bibnamefont {Preskill}},\ }\href
  {https://doi.org/10.1038/s41567-020-0932-7} {\bibfield  {journal} {\bibinfo
  {journal} {Nature Physics}\ }\textbf {\bibinfo {volume} {16}},\ \bibinfo
  {pages} {1050} (\bibinfo {year} {2020})}\BibitemShut {NoStop}%
\bibitem [{\citenamefont {Caponnetto}\ and\ \citenamefont
  {De~Vito}(2007)}]{CaponnettoDeVito2007}%
  \BibitemOpen
  \bibfield  {author} {\bibinfo {author} {\bibfnamefont {A.}~\bibnamefont
  {Caponnetto}}\ and\ \bibinfo {author} {\bibfnamefont {E.}~\bibnamefont
  {De~Vito}},\ }\href {https://doi.org/10.1007/s10208-006-0196-8} {\bibfield
  {journal} {\bibinfo  {journal} {Foundations of Computational Mathematics}\
  }\textbf {\bibinfo {volume} {7}},\ \bibinfo {pages} {331} (\bibinfo {year}
  {2007})}\BibitemShut {NoStop}%
\bibitem [{\citenamefont {Huang}\ \emph {et~al.}(2022)\citenamefont {Huang},
  \citenamefont {Broughton}, \citenamefont {Cotler}, \citenamefont {Chen},
  \citenamefont {Li}, \citenamefont {Mohseni}, \citenamefont {Neven},
  \citenamefont {Babbush}, \citenamefont {Kueng}, \citenamefont {Preskill},\
  and\ \citenamefont {McClean}}]{Huang2022Science}%
  \BibitemOpen
  \bibfield  {author} {\bibinfo {author} {\bibfnamefont {H.-Y.}\ \bibnamefont
  {Huang}}, \bibinfo {author} {\bibfnamefont {M.}~\bibnamefont {Broughton}},
  \bibinfo {author} {\bibfnamefont {J.}~\bibnamefont {Cotler}}, \bibinfo
  {author} {\bibfnamefont {S.}~\bibnamefont {Chen}}, \bibinfo {author}
  {\bibfnamefont {J.}~\bibnamefont {Li}}, \bibinfo {author} {\bibfnamefont
  {M.}~\bibnamefont {Mohseni}}, \bibinfo {author} {\bibfnamefont
  {H.}~\bibnamefont {Neven}}, \bibinfo {author} {\bibfnamefont
  {R.}~\bibnamefont {Babbush}}, \bibinfo {author} {\bibfnamefont
  {R.}~\bibnamefont {Kueng}}, \bibinfo {author} {\bibfnamefont
  {J.}~\bibnamefont {Preskill}},\ and\ \bibinfo {author} {\bibfnamefont
  {J.~R.}\ \bibnamefont {McClean}},\ }\href
  {https://doi.org/10.1126/science.abn7293} {\bibfield  {journal} {\bibinfo
  {journal} {Science}\ }\textbf {\bibinfo {volume} {376}},\ \bibinfo {pages}
  {1182} (\bibinfo {year} {2022})}\BibitemShut {NoStop}%
\bibitem [{\citenamefont {Liu}\ \emph {et~al.}(2025)\citenamefont {Liu},
  \citenamefont {Brunel}, \citenamefont {{\O}stergaard}, \citenamefont
  {Cordero}, \citenamefont {Chen}, \citenamefont {Wong}, \citenamefont
  {Nielsen}, \citenamefont {Bregnsbo}, \citenamefont {Zhou}, \citenamefont
  {Huang}, \citenamefont {Oh}, \citenamefont {Jiang}, \citenamefont {Preskill},
  \citenamefont {Neergaard-Nielsen},\ and\ \citenamefont
  {Andersen}}]{Liu2025Science}%
  \BibitemOpen
  \bibfield  {author} {\bibinfo {author} {\bibfnamefont {Z.-H.}\ \bibnamefont
  {Liu}}, \bibinfo {author} {\bibfnamefont {R.}~\bibnamefont {Brunel}},
  \bibinfo {author} {\bibfnamefont {E.~E.~B.}\ \bibnamefont {{\O}stergaard}},
  \bibinfo {author} {\bibfnamefont {O.}~\bibnamefont {Cordero}}, \bibinfo
  {author} {\bibfnamefont {S.}~\bibnamefont {Chen}}, \bibinfo {author}
  {\bibfnamefont {Y.}~\bibnamefont {Wong}}, \bibinfo {author} {\bibfnamefont
  {J.~A.~H.}\ \bibnamefont {Nielsen}}, \bibinfo {author} {\bibfnamefont
  {A.~B.}\ \bibnamefont {Bregnsbo}}, \bibinfo {author} {\bibfnamefont
  {S.}~\bibnamefont {Zhou}}, \bibinfo {author} {\bibfnamefont {H.-Y.}\
  \bibnamefont {Huang}}, \bibinfo {author} {\bibfnamefont {C.}~\bibnamefont
  {Oh}}, \bibinfo {author} {\bibfnamefont {L.}~\bibnamefont {Jiang}}, \bibinfo
  {author} {\bibfnamefont {J.}~\bibnamefont {Preskill}}, \bibinfo {author}
  {\bibfnamefont {J.~S.}\ \bibnamefont {Neergaard-Nielsen}},\ and\ \bibinfo
  {author} {\bibfnamefont {U.~L.}\ \bibnamefont {Andersen}},\ }\href
  {https://doi.org/10.1126/science.adv2560} {\bibfield  {journal} {\bibinfo
  {journal} {Science}\ }\textbf {\bibinfo {volume} {389}},\ \bibinfo {pages}
  {1332} (\bibinfo {year} {2025})}\BibitemShut {NoStop}%
\bibitem [{\citenamefont {King}\ \emph {et~al.}(2024)\citenamefont {King},
  \citenamefont {Wan},\ and\ \citenamefont {McClean}}]{King2024}%
  \BibitemOpen
  \bibfield  {author} {\bibinfo {author} {\bibfnamefont {R.}~\bibnamefont
  {King}}, \bibinfo {author} {\bibfnamefont {K.}~\bibnamefont {Wan}},\ and\
  \bibinfo {author} {\bibfnamefont {J.~R.}\ \bibnamefont {McClean}},\ }\href
  {https://doi.org/10.1103/PRXQuantum.5.040301} {\bibfield  {journal} {\bibinfo
   {journal} {PRX Quantum}\ }\textbf {\bibinfo {volume} {5}},\ \bibinfo {pages}
  {040301} (\bibinfo {year} {2024})}\BibitemShut {NoStop}%
\bibitem [{\citenamefont {Valiant}(1984)}]{Valiant1984}%
  \BibitemOpen
  \bibfield  {author} {\bibinfo {author} {\bibfnamefont {L.~G.}\ \bibnamefont
  {Valiant}},\ }\href {https://doi.org/10.1145/1968.1972} {\bibfield  {journal}
  {\bibinfo  {journal} {Communications of the ACM}\ }\textbf {\bibinfo {volume}
  {27}},\ \bibinfo {pages} {1134} (\bibinfo {year} {1984})}\BibitemShut
  {NoStop}%
\bibitem [{\citenamefont {Angluin}\ and\ \citenamefont
  {Laird}(1988)}]{AngluinLaird1988}%
  \BibitemOpen
  \bibfield  {author} {\bibinfo {author} {\bibfnamefont {D.}~\bibnamefont
  {Angluin}}\ and\ \bibinfo {author} {\bibfnamefont {P.}~\bibnamefont
  {Laird}},\ }\href {https://doi.org/10.1023/A:1022873112823} {\bibfield
  {journal} {\bibinfo  {journal} {Machine Learning}\ }\textbf {\bibinfo
  {volume} {2}},\ \bibinfo {pages} {343} (\bibinfo {year} {1988})}\BibitemShut
  {NoStop}%
\bibitem [{\citenamefont {McClean}\ \emph {et~al.}(2018)\citenamefont
  {McClean}, \citenamefont {Boixo}, \citenamefont {Smelyanskiy}, \citenamefont
  {Babbush},\ and\ \citenamefont {Neven}}]{McClean2018}%
  \BibitemOpen
  \bibfield  {author} {\bibinfo {author} {\bibfnamefont {J.~R.}\ \bibnamefont
  {McClean}}, \bibinfo {author} {\bibfnamefont {S.}~\bibnamefont {Boixo}},
  \bibinfo {author} {\bibfnamefont {V.~N.}\ \bibnamefont {Smelyanskiy}},
  \bibinfo {author} {\bibfnamefont {R.}~\bibnamefont {Babbush}},\ and\ \bibinfo
  {author} {\bibfnamefont {H.}~\bibnamefont {Neven}},\ }\href
  {https://doi.org/10.1038/s41467-018-07090-4} {\bibfield  {journal} {\bibinfo
  {journal} {Nature Communications}\ }\textbf {\bibinfo {volume} {9}},\
  \bibinfo {pages} {4812} (\bibinfo {year} {2018})}\BibitemShut {NoStop}%
\bibitem [{\citenamefont {Cerezo}\ \emph {et~al.}(2021)\citenamefont {Cerezo},
  \citenamefont {Sone}, \citenamefont {Volkoff}, \citenamefont {Cincio},\ and\
  \citenamefont {Coles}}]{Cerezo2021BP}%
  \BibitemOpen
  \bibfield  {author} {\bibinfo {author} {\bibfnamefont {M.}~\bibnamefont
  {Cerezo}}, \bibinfo {author} {\bibfnamefont {A.}~\bibnamefont {Sone}},
  \bibinfo {author} {\bibfnamefont {T.}~\bibnamefont {Volkoff}}, \bibinfo
  {author} {\bibfnamefont {{\L}.}~\bibnamefont {Cincio}},\ and\ \bibinfo
  {author} {\bibfnamefont {P.~J.}\ \bibnamefont {Coles}},\ }\href
  {https://doi.org/10.1038/s41467-021-21728-w} {\bibfield  {journal} {\bibinfo
  {journal} {Nature Communications}\ }\textbf {\bibinfo {volume} {12}},\
  \bibinfo {pages} {1791} (\bibinfo {year} {2021})}\BibitemShut {NoStop}%
\bibitem [{\citenamefont {Wang}\ \emph {et~al.}(2021)\citenamefont {Wang},
  \citenamefont {Fontana}, \citenamefont {Cerezo}, \citenamefont {Sharma},
  \citenamefont {Sone}, \citenamefont {Cincio},\ and\ \citenamefont
  {Coles}}]{Wang2021Noise}%
  \BibitemOpen
  \bibfield  {author} {\bibinfo {author} {\bibfnamefont {S.}~\bibnamefont
  {Wang}}, \bibinfo {author} {\bibfnamefont {E.}~\bibnamefont {Fontana}},
  \bibinfo {author} {\bibfnamefont {M.}~\bibnamefont {Cerezo}}, \bibinfo
  {author} {\bibfnamefont {K.}~\bibnamefont {Sharma}}, \bibinfo {author}
  {\bibfnamefont {A.}~\bibnamefont {Sone}}, \bibinfo {author} {\bibfnamefont
  {{\L}.}~\bibnamefont {Cincio}},\ and\ \bibinfo {author} {\bibfnamefont
  {P.~J.}\ \bibnamefont {Coles}},\ }\href
  {https://doi.org/10.1038/s41467-021-27045-6} {\bibfield  {journal} {\bibinfo
  {journal} {Nature Communications}\ }\textbf {\bibinfo {volume} {12}},\
  \bibinfo {pages} {6961} (\bibinfo {year} {2021})}\BibitemShut {NoStop}%
\bibitem [{\citenamefont {Anschuetz}\ \emph {et~al.}(2023)\citenamefont
  {Anschuetz}, \citenamefont {Hu}, \citenamefont {Huang},\ and\ \citenamefont
  {Gao}}]{Anschuetz2023}%
  \BibitemOpen
  \bibfield  {author} {\bibinfo {author} {\bibfnamefont {E.~R.}\ \bibnamefont
  {Anschuetz}}, \bibinfo {author} {\bibfnamefont {H.-Y.}\ \bibnamefont {Hu}},
  \bibinfo {author} {\bibfnamefont {J.-L.}\ \bibnamefont {Huang}},\ and\
  \bibinfo {author} {\bibfnamefont {X.}~\bibnamefont {Gao}},\ }\href
  {https://doi.org/10.1103/PRXQuantum.4.020338} {\bibfield  {journal} {\bibinfo
   {journal} {PRX Quantum}\ }\textbf {\bibinfo {volume} {4}},\ \bibinfo {pages}
  {020338} (\bibinfo {year} {2023})}\BibitemShut {NoStop}%
\bibitem [{\citenamefont {Cortes}\ and\ \citenamefont
  {Vapnik}(1995)}]{CortesVapnik1995}%
  \BibitemOpen
  \bibfield  {author} {\bibinfo {author} {\bibfnamefont {C.}~\bibnamefont
  {Cortes}}\ and\ \bibinfo {author} {\bibfnamefont {V.}~\bibnamefont
  {Vapnik}},\ }\href {https://doi.org/10.1007/BF00994018} {\bibfield  {journal}
  {\bibinfo  {journal} {Machine Learning}\ }\textbf {\bibinfo {volume} {20}},\
  \bibinfo {pages} {273} (\bibinfo {year} {1995})}\BibitemShut {NoStop}%
\end{thebibliography}

\end{document}